%% file: main.tex
\documentclass[a4paper,left=1.5cm,right=1.5cm,top=2cm,bottom=2cm,onecolumn,11pt,accepted=2024-11-08]{quantumarticle}
\pdfoutput=1

\RequirePackage{etex}

\usepackage[dvipsnames,table]{xcolor}

\usepackage{authblk}
\usepackage[utf8]{inputenc}
\DeclareUnicodeCharacter{0301}{\'{e}} 

\usepackage{wrapfig}

\usepackage{hyperref,physics}
\usepackage{url}
\usepackage{ifthen}
\usepackage{amsmath,amsthm,amssymb}
\usepackage{verbatim}
\usepackage{array}
\usepackage{mathtools}

\DeclarePairedDelimiter\floor{\lfloor}{\rfloor}

\usepackage{multirow}

\usepackage{enumitem}

\usepackage[export]{adjustbox}

\usepackage{soul}
\usepackage[normalem]{ulem}

\newif\ifreviewmode
\reviewmodefalse
\newcommand{\add}[1]{\ifreviewmode\textcolor{blue}{#1}\else#1\fi}
\newcommand{\remove}[1]{\ifreviewmode\textcolor{red}{\sout{#1}}\else\fi}
\newcommand{\replace}[2]{\remove{#1}\add{#2}}

\usepackage{subcaption}

\usepackage{algorithm}
\usepackage{algpseudocode}

\usepackage{lipsum}

\usepackage{nicematrix}

\usepackage[nowidering]{yhmath}

\newcommand{\overbar}[1]{\mkern 1.5mu\overline{\mkern-1.5mu#1\mkern-1.5mu}\mkern 1.5mu}

\newcommand{\N}[0]{\mathbb{N}}

\newcommand{\R}[0]{\mathbb{R}}
\newcommand{\C}[0]{\mathbb{C}}

\newcommand{\Span}[1]{\operatorname{span}\!\left(#1\right)}
\renewcommand{\Trace}[1]{\operatorname{tr}\left(#1\right)}

\newcommand{\lX}[0]{\overbar{X}}

\newcommand\restr[2]{{%
    \left.\kern-\nulldelimiterspace %
    #1 %
    \vphantom{\big|} %
    \right|_{#2} %
}}

\newcommand{\noinitial}[1]{}

\definecolor{SO}{HTML}{e66101}
\definecolor{SOR}{HTML}{fdb863}
\definecolor{ST}{HTML}{5e3c99}
\definecolor{STR}{HTML}{b2abd2}

\usepackage{tikz,pgfplots}
\pgfplotsset{compat = newest}
\usetikzlibrary{pgfplots.groupplots}
\usetikzlibrary{graphs, graphs.standard}
\usepackage{tikz-network}

\usetikzlibrary{quantikz}
\usetikzlibrary{tikzmark}

\usetikzlibrary{fadings,shapes.arrows,shadows,intersections}

\tikzfading[name=arrowfading, top color=transparent!0, bottom color=transparent!95]
\tikzset{arrowfill/.style={top color=black!2, bottom color=black!80, general shadow={fill=black, shadow yshift=-0.8ex, path fading=arrowfading}}}
\tikzset{arrowstyle/.style={draw=black,arrowfill, double arrow,minimum height=#1, double arrow,
single arrow head extend=.4cm,}}

\usepackage{booktabs}

\usepackage{simplewick}     

\usepackage[
backend=biber,
style=numeric,
url=false,
doi=true
]{biblatex}

\usepackage[rightmargin=2cm]{thmbox}
\newtheorem[L]{theorem}{Theorem}
\newtheorem[L]{proposition}{Proposition}
\newtheorem[L]{lemma}{Lemma}
\newtheorem[L]{remark}{Remark}
\newtheorem[L]{corollary}{Corollary}
\newtheorem[L]{problem}{Problem}
\newtheorem[L]{condition}{Condition}
\newtheorem[L]{example}{Example}
\newtheorem[L]{definition}{Definition}

\title{LX-mixers for QAOA: Optimal mixers restricted to subspaces and the stabilizer formalism}
\date{14.11.2024}

\author{Franz G. Fuchs}
\affiliation{SINTEF AS, Department of Mathematics and Cybernetics, Oslo, Norway}
\affiliation{Department of Mathematics, University of Oslo, Norway}
\orcid{0000-0003-3558-503X}
\author{Ruben Pariente Bassa}
\affiliation{SINTEF AS, Department of Mathematics and Cybernetics, Oslo, Norway}

\usepackage{capt-of,etoolbox}

\definecolor{IXII}{RGB}{58.6004535,76.17308133,192.189204015}
\definecolor{IIXX}{RGB}{170.14949570500002,198.68999653499998,253.204599315}
\definecolor{XIXI}{RGB}{246.891866745,183.8152455,156.13471279}
\definecolor{XIII}{RGB}{179.94665529,3.9668208,38.30936706}

\begin{document}
        
\maketitle

\abstract{
We present a novel formalism to both understand and construct mixers that preserve a given subspace.
The method connects and utilizes the stabilizer formalism that is used in error correcting codes.
This can be useful in the setting when the 
quantum approximate optimization algorithm (QAOA), a popular meta-heuristic for solving combinatorial optimization problems, is applied in the setting where the constraints of the problem lead to a feasible subspace that is large but easy to specify.
The proposed method gives a systematic way to construct mixers that are resource efficient in the number of controlled not gates and can be understood as a generalization of the well-known X- and XY-mixers and a relaxation of the Grover mixer: 
Given a basis of any subspace, a resource efficient mixer can be constructed that preserves the subspace.
The numerical examples provided show a dramatic reduction of CX gates when compared to previous results.
We call our approach logical X-Mixer or logical X QAOA (\textbf{LX-QAOA}), since it can be understood as dividing the subspace into code spaces of stabilizers S and
consecutively applying logical rotational X gates associated with these code spaces.
Overall, we hope that this new perspective can lead to further insight into the development of quantum algorithms.
}
  \begin{center}
  \captionsetup{hypcap=false}

        \begin{tikzpicture}[scale=.3]

        \begin{scope}[fill opacity=0.5,text opacity=1]

        \definecolor{IXII}{RGB}{58.6004535,76.17308133,192.189204015};
        \definecolor{IIXX}{RGB}{170.14949570500002,198.68999653499998,253.204599315};
        \definecolor{XIXI}{RGB}{246.891866745,183.8152455,156.13471279};
        \definecolor{XIII}{RGB}{179.94665529,3.9668208,38.30936706};
        
        \def\sx{-16}
        
        \node[] at (\sx+8,0) {\phantom{$\cdots \longrightarrow$}};
        
        \def\sx{0}
        
        \draw[draw = black] (\sx,0) circle (6);
        \draw[fill=gray, draw = black,name path=circle 1] (\sx,0) circle (4);
        \node[] at (\sx,6.7) {$\C^{2^n}$};
        \node[] at (\sx,4.5) {$\Span{B}$};
        
        \node [
        ellipse,
        minimum width=.3cm,
        minimum height=0.625cm,
        draw,
        fill=XIXI, fill opacity=0.25] (V11) at (\sx-2.5,-1) {};
        
        \node [
        ellipse,
        minimum width=.3cm,
        minimum height=0.625cm,
        draw,
        fill=XIXI, fill opacity=0.75] (V12) at (\sx+0.5,-1) {};
        
        \draw [thick, -latex] (V11) to [bend right=45] (V12);
        \draw [thick, -latex] (V12) to [bend right=45] (V11);

        \node[] at (\sx+8,0) {$\longrightarrow$};
        
        \def\sx{16}
        
        \draw[draw = black] (\sx,0) circle (6);
        \draw[fill=gray, draw = black,name path=circle 1] (\sx,0) circle (4);
        
        \node [
        ellipse,
        minimum height=.5cm,
        minimum width=.75cm,
        draw,
        fill=XIII, fill opacity=.25] (V21) at (\sx-.6,-1) {};
        
        \node [
        ellipse,
        minimum height=.5cm,
        minimum width=.75cm,
        draw,
        fill=XIII, fill opacity=.75] (V22) at (\sx-.6,2) {};
        
        \draw [thick, -latex] (V21) to [bend right=45] (V22);
        \draw [thick, -latex] (V22) to [bend right=45] (V21);
        
        \node[] at (\sx,6.7) {$\C^{2^n}$};
        \node[] at (\sx,4.5) {$\Span{B}$};

        \node[] at (\sx+8,0) {$\longrightarrow$};
        
        \def\sx{32}
        
        \draw[draw = black] (\sx,0) circle (6);
        \draw[fill=gray, draw = black,name path=circle 1] (\sx,0) circle (4);
        
        \node [
        ellipse,
        minimum height=.75cm,
        minimum width=.5cm,
        draw,
        fill=IXII, fill opacity=.25] (V31) at (\sx-1,1.5) {};
        
        \node [
        ellipse,
        minimum height=.75cm,
        minimum width=.5cm,
        draw,
        fill=IXII, fill opacity=.75] (V32) at (\sx+2.,1.5) {};
        
        \draw [thick, -latex] (V31) to [bend right=45] (V32);
        \draw [thick, -latex] (V32) to [bend right=45] (V31);
        
        \node[] at (\sx,6.7) {$\C^{2^n}$};
        \node[] at (\sx,4.5) {$\Span{B}$};

        \node[] at (\sx+8,0) {$\longrightarrow \cdots$};
        
        \end{scope}
        \end{tikzpicture}


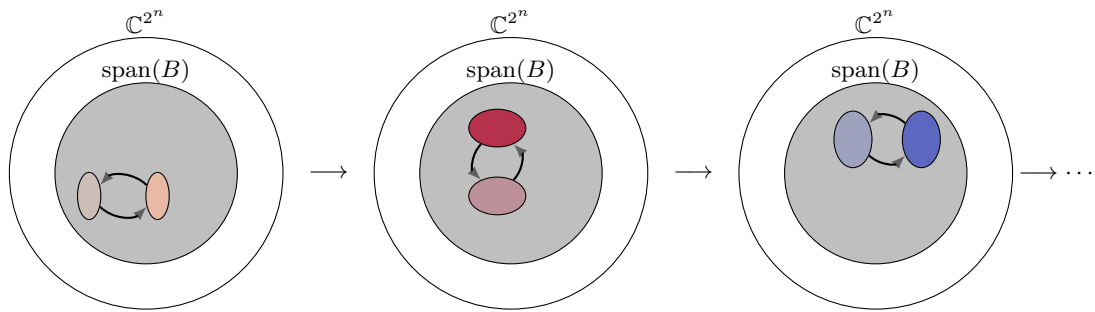
\captionof{figure}{
Schematic view of the proposed method to construct efficient mixers for the quantum alternating operator ansatz (QAOA) constrained to a subspace.
The method can be understood as dividing the subspace into code spaces of stabilizers S and
consecutively applying logical rotational X gates associated with the code spaces.
}
\label{fig:teaser}
  \captionsetup{hypcap=true}
  \end{center}

\section{Introduction}
The quantum approximate optimization algorithm (QAOA)~\cite{farhi2014quantum,hadfield2019quantum} is a popular meta-heuristic for solving combinatorial optimization problems in the form of a quadratic (unconstrained) binary optimization (QUBO) problem.
There is both evidence that points towards an advantage over classical algorithms~\cite{diez2023quantum} as well as against it~\cite{stilck2021limitations}, particularly in the NISQ era.
Given an objective function $f:\{0,1\}^n \rightarrow \R$ one defines the problem Hamiltonian through $H_P\ket{x} = f(x) \ket{x}$.
Consequently, ground states of $H_P$ minimize the objective function.
The general form of QAOA is to prepare a parametrized quantum state given by
\begin{equation}
    \ket{\gamma,\beta} = U_M(\beta_p) U_P(\gamma_p)\cdots U_M(\beta_1) U_P(\gamma_1) \ket{\phi_0}.
\end{equation}
That is, starting with an initial state $\ket{\phi_0}$ one alternates the application of a phase separating operator $U_P(\gamma)$ and mixing operator $U_M(\beta)$.
The goal of QAOA is then to find optimal parameters $\gamma,\beta \in \R$, such that $\bra{\gamma,\beta} H_P \ket{\gamma,\beta}$ is minimized.
Several extensions and variants of QAOA have been proposed since its inception. One such variant is ADAPT-QAOA~\cite{zhu_adaptive_2020}, which is an iterative version tailored to specific hardware constraints. Another variant is R-QAOA~\cite{bravyi2019obstacles}, which recursively removes variables from the Hamiltonian until the remaining instance can be solved classically. A third variant is WS-QAOA~\cite{eggerwarmstarting2021}, which takes into account solutions of classical algorithms to warm-start QAOA.

In this article we consider the case where solutions are constrained to a feasible subspace $\Span{B}\subset \mathcal{H}$ of the full Hilbert space $\mathcal{H}=(\C^2)^{\otimes n}$. The subspace is described by a subset of all computational basis states, i.e.,
\begin{equation}
    B=\left\{\ket*{z_j}, \ \ 1\leq j\leq J, \ \ z_j\in\{0,1\}^n \right\}.%
\end{equation}
There are \replace{two main}{multiple} ways to take constraints into account with QAOA. One popular approach is to penalize unsatisfied constraints in the objective function by formulating them as another QUBO. Although it is straightforward to define a phase separating Hamiltonian for the resulting QUBO, this approach has some downsides. In particular, the probability of obtaining infeasible solutions is typically rather high, which reduces efficiency. Additionally, the quality of the results is very sensitive to the penalization weight, which must be carefully balanced between optimality and feasibility~\cite{Wang2020}. This can be a challenging task since the optimal weight depends on the specific problem and there is a trade-off between optimality and feasibility in the objective function. Applications of this approach include graph coloring problems, the traveling salesperson problem~\cite{hadfield2017quantum}, the Max-k-cut problem~\cite{fuchs2021efficient}, portfolio optimization~\cite{brandhofer2023benchmarking}, and the tail-assignment problem~\cite{PhysRevApplied.14.034009}.
\add{Other approaches have been proposed recently that do not rely on the penalization of the cost function but for example on postprocessing the unfeasible states~\cite{10472069}.
Solving combinatorial optimizations problems with hard constraints can also be done by combining the quantum variational framework with a genetic algorithm for multiobjective optimization~\cite{D_ez_Valle_2023}.
Another alternative is to perform the optimization considering only in-constraint samples obtained from a non-constraint preserving ansatz~\cite{Hao_2022}.
Finally, quantum Zeno dynamics can be employed to restrict the wave function to the feasible subspace via repeated projective measurements using ancillary qubits~\cite{Herman_2023}.}

\replace{The second approach}{The approach presented in this work} to model constraints with QAOA is to design mixers that have zero probability to transition from a feasible state to an infeasible one, thereby rendering the penalization weight from the previous approach unnecessary. However, creating mixers that account for constraints is generally more difficult.
\add{In this work, we define the mixer for feasible states using logical operations, drawing an analogy to quantum error correction. This approach allows us to create a structure termed the logical X-Mixer (\textbf{LX-Mixer}). The schematic representation of our construction is illustrated in Figure~\ref{fig:teaser}, where each logical X rotation facilitates mixing exclusively between codespaces within the feasible subspace.}
The $XY$-mixer \cite{hadfield2017quantum,Wang2020}
is a well-known example in the literature 
\remove{which constrains the evolution to states with non-zero overlap with "k-hot" states, i.e., computational basis states with exactly k entries equal to one.}
\add{which confines the evolution to the subspace of ``k-hot'' states, i.e., subspace spanned by computational basis states with exactly k entries equal to one.
This mixer, which can be described in terms of logical X operations, has been successfully implemented for constrained optimization problems on real quantum devices \cite{Niroula_2022}.}
\remove{The reason for this is that $H = X_iX_j + Y_iY_j$ is a number preserving operator.}

One notable example is GM-QAOA~\cite{bartschi2020grover} which uses an ansatz of the form in Equation~\eqref{eq:Hmdefinition} with $T_{j,k}=1, \forall j, k$, i.e. it is special case of our approach, providing all-to-all mixing of states. An implementation can be based on the Grover-like form $ U_M(\beta) = U_S \left(I + (e^{-i\beta}-1)\ket{0}\bra{0} \right)U_S^\dagger$, where $U_S$ is a circuit creating equal superposition of all feasible states, and $I + (e^{-i\beta}-1)\ket{0}\bra{0}$ is implemented with a rotational gate with $n-1$ controls.
\add{A similar mixer construction, based on continuous quantum walks, that provides mixing for cases of circulant adjacency matrices $T_{jk}$ is proposed in \cite{Wang}.}
However, it might be advantageous to use the full flexibility of choosing $T$ to create valid mixers.
In general, a mixer is called valid if the following conditions hold~\cite{hadfield2019quantum}.
\begin{definition}[Valid mixer]
\label{def:validmixer}
A mixer is called valid if it 
 \textit{preserves the feasible subspace}, i.e.,
    \begin{equation}
    \label{eq:mixer_preserve}
        U_M(\beta)\ket{v} \in \Span{B}, \quad \forall \ket{v} \in \Span{B}, \forall \beta\in\R,%
    \end{equation}
and if it \textit{provides transitions between all pairs of feasible states}, i.e., for each pair of computational basis states $\ket{x}, \ket{y} \in B$ there exist
    $\beta^*\in\R$ and $r \in \N\cup\{0\}$, such that 
    \begin{equation}
    \label{eq:mixer_transition}
        |\bra{x}
        \underbrace{U_M(\beta^*)\cdots U_M(\beta^*)}_{r\text{ times}}
        \ket{y}|
            >0.
    \end{equation}
\end{definition}

However, in this article we choose a different ansatz of the form~\cite{fuchs2022constraint}:
\begin{equation}
    U_M(\beta) = e^{-i\beta H_M},
    \quad H_M = \sum_{j<k} (T)_{j,k} H_{z_j\leftrightarrow z_k},
    \quad H_{z_j\leftrightarrow z_k} = \ket*{z_j}\!\!\bra*{z_k} + \ket*{z_k}\!\!\bra*{z_j}.
    \label{eq:Hmdefinition}
\end{equation}

The space of bounded linear operators $\mathcal{B}(\mathcal{H})$ is endowed with the Hilbert-Schmidt inner product $(A,B)=\Trace{A^\dagger B}$~\cite{siewert2022orthogonal}.
The operator $H_M$ in Equation~\eqref{eq:Hmdefinition} is expressed on the standard/computational basis $\{\ket{i}\bra{j}\}, i,j=0,\dots, 2^n-1$.
In order to allow for a practical implementation we express $H_M$ in the Pauli basis $P_j\in \{ I,X,Y,Z \}^n$, which can be achieved via
\begin{equation}\label{eq:Paulidecomp}
    H_M = \sum_j w_j P_j, \quad w_j=\frac{1}{2^n} \Trace{ P_j^\dagger H_M }\in\R.
\end{equation}

After using a standard Trotterization scheme~\cite{hatano2005finding,trotter1959product} (which is exact for commuting Paul-strings),
\begin{equation}
    U_M(t) = e^{-it H_M} \approx \prod_{\underset{|w_j|>0}{j}} e^{-it w_j P_j},
    \label{eq:Trotterization}
\end{equation}
each term with non-zero weight can be implemented using a standard circuit
consisting of $2 (\operatorname{weight}(P_j)-1)$ controlled not gates $CX$, see Equation~\eqref{eq:eP}.
Lower depth variants to realize the circuit are possible, but the number of $CX$ gates is the same. Here, 
$0 \leq \operatorname{weight}(P_j) \leq n$ is the weight of a Pauli string $P_j$, which is the number of qubits that $P_j$ acts on by a nontrivial Pauli matrix, i.e. $X,Y$ or $Z$.

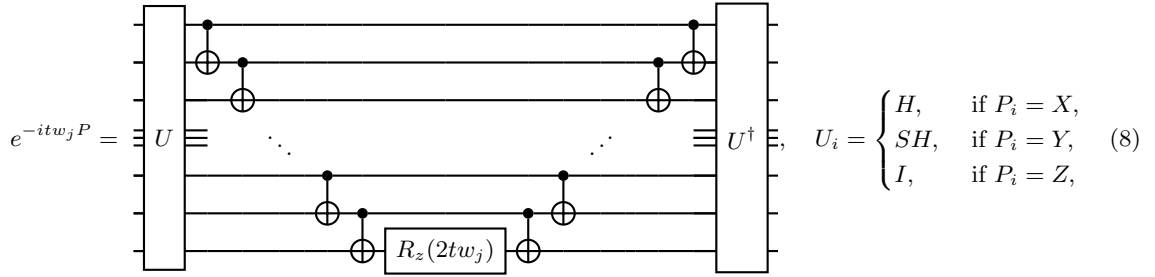
\begin{equation}
    e^{-i t w_j P} = 
    \begin{quantikz}[row sep={0.5cm,between origins},column sep=1ex]
    \qw&\gate[wires=7]{U}
    & \ctrl{1} & \qw & \qw &\qw& \qw& \qw& \qw & \qw & \qw &\qw& \qw& \qw& \qw& \ctrl{1} & \gate[wires=7]{U^\dagger}& \qw\\
    \qw&\qw& \targ{} &\ctrl{1} & \qw & \qw& \qw& \qw& \qw& \qw & \qw & \qw& \qw& \qw& \ctrl{1} &\targ{}&\qw& \qw\\
    \qw& \qw& \qw \qw & \targ{} & \qw & \qw & \qw& \qw& \qw& \qw & \qw & \qw& \qw& \qw&\targ{} & \qw & \qw & \qw\\
    &\qwbundle[alternate]{}&\qwbundle[alternate]{} &&\ddots&&&&&&&\adots&&&&&\qwbundle[alternate]{}&\qwbundle[alternate]{}\\
    \qw&\qw& \qw & \qw & \qw & \qw & \ctrl{1} & \qw & \qw & \qw & \ctrl{1}& \qw & \qw & \qw & \qw & \qw & \qw& \qw \\
    \qw&\qw& \qw & \qw & \qw & \qw & \targ{} & \ctrl{1} & \qw& \ctrl{1} &\targ{} & \qw & \qw & \qw & \qw & \qw & \qw& \qw \\
    \qw&\qw& \qw & \qw & \qw & \qw & \qw & \targ{} & \gate{R_z(2tw_j)} & \targ{}& \qw & \qw & \qw & \qw & \qw & \qw & \qw & \qw \\
    \end{quantikz},
    U_i=\begin{cases}
        H,& \hspace{-.4cm}\text{if } P_i=X,\\
        SH,& \hspace{-.4cm}\text{if } P_i=Y,\\
        I,& \hspace{-.4cm}\text{if } P_i=Z,
    \end{cases}
    \label{eq:eP}
    \end{equation}
where $S$ is the phase gate and $H$ is the Hadamard gate.
Hence, a good indicator for the cost of a quantum algorithm is given by the number of required $CX$ gates, defined as
\begin{equation}
    \operatorname{Cost}(U_M) = \operatorname{Cost}(H_M) = \sum_{\underset{|w_j|\neq 0}{P_j\neq I}}
    2 \left(\operatorname{weight}(P_j)-1\right).
    \label{eq:optimality}
\end{equation}

In this article we develop algorithms to find constraint preserving mixers that minimize the number of CX gates required.
Given $B$ there is a generally a lot of design freedom for $U_M$ to be a valid mixer, which can heavily impact the cost.%
    There are many matrices $T$ that lead to valid mixers. In Section~\ref{sec:T} we discuss how an optimal strategy to obtain the lowest cost is related to solving a graph optimization problem.
    In Section~\ref{sec:stabilizer} we make a connection between mixers and the stabilizer formalism. This can be utilized to a potentially large reduction of the cost.
    A further reduction of the cost comes when taking into account that the operators only have to act on the feasible subspace, see Section~\ref{sec:restrictingprojectors}
    Finally, we present examples of our approach in Section~\ref{sec:examples}.

\section{Optimal Trotterization and the adjacency matrix}
\label{sec:T}

\begin{figure}
    \centering
    \input{pictures/adjacency}
    \caption{The matrix $T$ in Equation~\eqref{eq:Hmdefinition} can be interpreted as the adjacency matrix of a graph $G_T$ and vice versa.
    For a mixer to be valid the graph must be connected.
    }
    \label{fig:adjacency}
\end{figure}
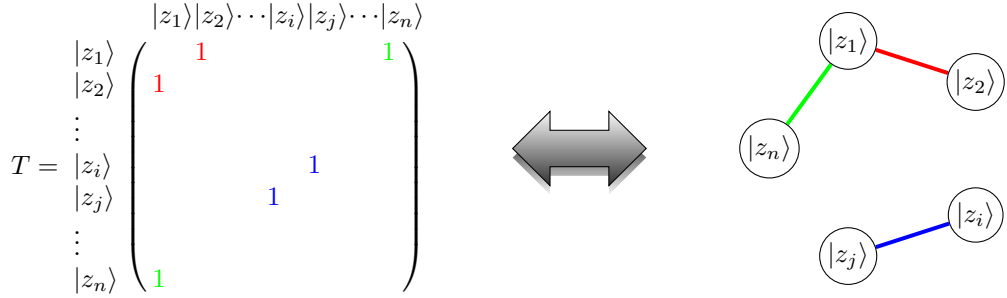

Mixers of the form given in Equation~\eqref{eq:Hmdefinition} are valid if for $T \in\R^{J\times J} $
we have that $T$ is symmetric and $\forall 1\leq j,k \leq J, \exists r_{j,k}\in\N\cup\{0\}$, s.t. $(T^{r_{j,k}})_{j,k} \neq 0$~\cite[Theorem~1]{fuchs2022constraint}.
This condition is not very intuitive, is cumbersome to check and it is not clear how one would practically find optimal mixers using it.
In this section we will describe a practical way to construct a valid Trotterization along with a condition of optimality.

\subsection{Valid Trotterization}
The matrix $T$ in Equation~\eqref{eq:Hmdefinition} can be interpret as the adjacency matrix of a graph $G_T$, see Figure~\ref{fig:adjacency}. 
A mixer provides transitions between all pairs of feasible states if $G_T$ if a connected graph.

\begin{theorem}[Valid mixer]
\label{theorem:adjacencygraphconnected}
The mixer Hamiltonian $H_M$ defined in Equation~\eqref{eq:Hmdefinition} is valid according to Definition~\ref{def:validmixer} if the graph $G_T$ of the adjacency matrix $T$ is undirected and connected.
\end{theorem}
\begin{proof}
We show equivalence with~\cite[Theorem~1]{fuchs2022constraint}.
By definition, T is symmetric for an undirected graph. 
If T is the adjacency matrix of the undirected graph $G_T$, then the entry $(T^k)_{i,j}$ gives the number of undirected walks of length k from vertex i to vertex j.
Hence, the non-zero entries of $T^k$ gives all pairs of nodes that are connected by a path of length k.
From this we see that the conditions in~\cite[Theorem~1]{fuchs2022constraint} are fulfilled if $T$ is a connected graph.
\end{proof}

In general, the terms in Equation~\eqref{eq:Hmdefinition} will not commute and one needs to perform a Trotterization
by decomposing $T$ into a sum of symmetric matrices $T_q$,
where for each $q$ all non-zero Pauli-operators of $\sum_{j<k} (T_q)_{j,k} H_{z_j\leftrightarrow z_k}$ have to commute.
In particular, a Trotterized mixer has to fulfill that for all $1\leq j,k \leq |J|$ 
there exist $r_m\in\N\cup\{0\}$ (possibly depending on the pair) such that
$\Big(\prod_{\underset{q\in Q}{m=1}}^M T_{q}^{r_m}\Big)_{j,k} \neq 0$~\cite[Theorem~2]{fuchs2022constraint}.
From the proof of Theorem~\ref{theorem:adjacencygraphconnected} it follows that this condition is equivalent to the following.
\begin{corollary}[Valid Trotterized mixer]
\label{corollary:validtrotterization}
The mixer Hamiltonian $H_M$ defined in Equation~\eqref{eq:Hmdefinition} is valid according to Definition~\ref{def:validmixer} if the undirected graph $G_T=\cup_q G_{T_q}$ is connected
and for each $q$ all non-zero Pauli strings of the mixer $\sum_{j<k} (T_q)_{j,k} H_{z_j\leftrightarrow z_k}$ commute.
\end{corollary}
It is straight forward to check if a graph is connected.
What is still missing is a way to find a family of graphs where all non-zero Pauli strings commute.

\subsection{Optimal Trotterization}
How one can one identify all pairs of $B$ such that the corresponding operators $H_{v_k\leftrightarrow v_l}$ all commute?
We start by defining a \textit{logical} X operator $\mathbf{\lX}$
of a pair $\ket{x}\neq \ket{y} \in B$ through
\begin{equation}\label{eq:Xhat}
    \lX_i =
    \begin{cases}
        X,& \text{if } x_i \neq y_i\\
        I, & \text{otherwise}.
    \end{cases}
\end{equation}
The reason why we call it "logical" will become apparent when we make the connection to the stabilizer formalism in Section~\ref{sec:stabilizer}.
For all $\lX \in \{I,X\}^n\setminus\{I\}^n$ we define two sets: the set of all states in $B$ that are connected by applying a particular $\lX$, and pairs of all states that are connected, i.e.,
\begin{equation}
\label{eq:GraphPairs}
\begin{split}
    V_{\lX} &= \big\{ \ket{x} \in B \ \big| \ \exists \ket{y}\in B \ \text{s.t.} \  \lX \ket{x} = \ket{y} \big\}, \\
    E_{\lX} &= \big\{ \{\ket{x_i}, \ket{y_i}\}  \in B\times B  \ \big| \ \lX \ket{x} = \ket{y}\big\}.
\end{split}
\end{equation}
Note that $E_{\lX}$ is a set of \textit{unordered} pairs.
For the example given in Figure~\ref{fig:graph_minimization_problem} we have $V_{IXXX}=\{ \ket{1110},\allowbreak\ket{1001},\allowbreak\ket{0100},\allowbreak\ket{0011}\}$ and $E_{IXXX}=\{ \{\ket{1110},\ket{1001}\}, \{\ket{0100},\ket{0011}\}\}$.

\begin{definition}[Family of valid Mixers]
\label{definition:MixerFamily}
Given a feasible set $B$ we define the family of graphs $(G_{\lX})_{\lX \in \{I,X\}^n\setminus\{I\}^n}$ where $G_{\lX} = (B, E_{\lX})$.
This gives rise to a family of mixers $(H_{\lX})_{\lX \in \{I,X\}^n\setminus\{I\}^n}$ 
\begin{equation}
\label{eq:HlX}
H_{\lX}=
    \sum_{j<k} (T_{\lX})_{j,k} H_{z_j\leftrightarrow z_k} = \sum_{\{\ket{x},\ket{y}\}\in E_{\lX}} H_{x\leftrightarrow y},
\end{equation}
where $T_{\lX}$ is the adjacency matrix of the graph $G_{\lX}$.
For each $\lX$ all non-vanishing Pauli-strings of the mixer $H_{\lX}$ commute, see Corollary~\ref{corollary:commuting}.
\end{definition}
For each logical X operators we can write
\begin{equation}
\label{eq:Xprod}
\sum_{\{\ket{x},\ket{y}\}\in E_{\lX}} H_{x\leftrightarrow y}=
\sum_{\{\ket{x},\ket{y}\}\in E_{\lX}} \big(\ket{x}\bra{y}
+\ket{y}\bra{x} \big)
=
\lX \sum_{\ket{x}\in V_{\lX}} 
 \ket{x}\bra{x}
=
\lX \prod\nolimits_{\Span{ V_{\lX} }}.%
\end{equation}
The action of the Hamiltonian can therefore be understood as a projection onto a feasible subspace followed by a logical X operation.
Furthermore, observe that for any computational basis state $\ket{x}$ the operator $\ket{x}\bra{x}$ has basis in $\{I,Z\}^n$.

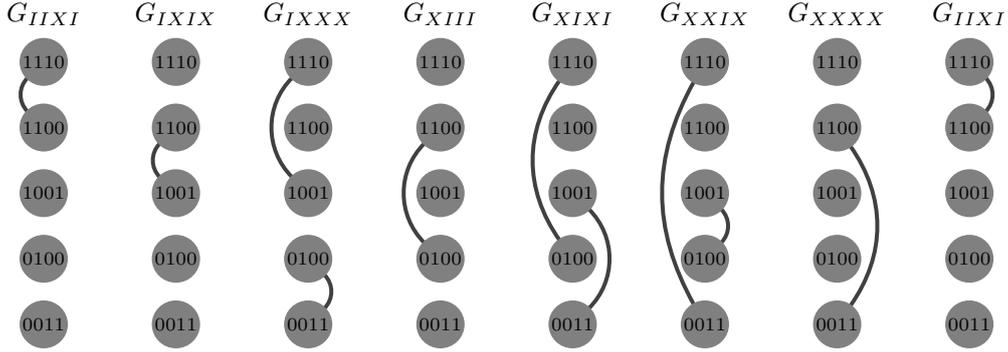
\begin{figure}
    \centering
    \input{pictures/graph_minimazation_problem}
    \caption{Family of graphs for $B=\{\ket{1110},\ket{1100},\ket{1001},\ket{0100},\ket{0011}\}$.
    A graph has an edge between two vertices $v_1, v_2$ if $\lX\ket{v_1} = \ket{v_2}$.
    Each graph has edges for the same operator $\lX$, which ensures that the Pauli-strings of the associated Hamiltonian all commute, see Equation~\eqref{eq:Xprod}.
    }
    \label{fig:graph_minimization_problem}
\end{figure}

All non-empty sets $E_{\lX}$ can be identified in time $\mathcal{O}(|B|^2)$, by going through all pairs of states in $B$ and grouping them together according to the logical X operator between them.
An example of such a set of graphs is given in Figure~\ref{fig:graph_minimization_problem}.

\begin{definition}[Optimal Trotterization]
\label{def:optimalTrotterization}
Let $(G_i)_i$ a set of graphs that contains all graphs according to Definition~\ref{definition:MixerFamily}, together with all possible subgraphs.
For each graph $G_i$ we define the cost $c_i=\sum_{\{\ket{x},\ket{y}\}\in E_i} \operatorname{Cost}(H_{x\leftrightarrow y})$.

\noindent
An optimal mixer is given by $\operatorname{argmin} \{c^T x \ \big| \ G=(B, \bigcup_{i|x_i=1} E_i) \text{ is a connected graph} \}$.
\end{definition}

In the next section we describe a way to efficiently calculate the mixer in the Pauli basis and a way to potentially reduce the cost of the projection operator by restricting its action to the feasible subspace $B$.

\section{Stabilizer formalism to construct mixers}
\label{sec:stabilizer}
In this section we show that mixers have a useful connection to the stabilizer formalism.
This becomes very convenient since it allows us to work with generating sets, instead of all elements of the group, meaning we can work with exponentially fewer elements. 

\subsection{Stabilizer formalism for a single pair of states}
\label{sec:stabilizeronepair}
For $\ket{x}\neq \ket{y}\in B$ one can interpret $C=\Span{\ket{x},\ket{y}}$ as the code space of a stabilizer $S$.
The stabilizer is defined as the commutative subgroup of $\mathcal{P}_n$, such that $S=\{g\in\mathcal{P}_n \ | \ g\ket*{z}=+\ket*{z} \ \forall \ket*{z}\in C\}$.
Here, $\mathcal{P}_n$ is the Pauli group on $n$ qubits.
Because $\operatorname{dim}(C)=2$ the stabilizer has a minimal generating set of $n-1$ elements, i.e., $S = \langle g_1, \cdots, g_{n-1}\rangle$. %
\add{The centralizer $\mathcal{C}$ of the stabilizer group S in $\mathcal{P}_n$ is the set of operators such that $\mathcal{C}=\{p\in\mathcal{P}_n \ | \ ps=sp \ \ \forall s\in S \} $.
The normalizer $\mathcal{N}$ of a stabilizer group consists of the set of operators such that $\mathcal{N}=\{p\in\mathcal{P}_n \ | \ psp^{\dag}\in S \ \ \forall s\in S \} $. Due to the fact that the stabilizer group does not contain $-I$, it turns out that the centralizer is equal to the normalizer.} 
Since $B$ consists of only computational basis states, i.e., strings of zeros and ones, it is very easy to find a generator for the stabilizer through the $\mathcal{O}(n)$ algorithm~\ref{alg:stabilizerC}.
As an example, for the pair $\ket{x}=\ket{10010}, \ket{y}=\ket{01011}$ a minimal generating set is given by $S=\langle Z_3, -Z_4, -Z_1Z_2, Z_2Z_5\rangle$. %
The sum of all elements of the stabilizer group acts as a projector onto the code subspace.
\begin{proposition}[Stabilizer subspace projector]
\label{prop:stabilizersubspaceprojector}
The stabilizer subspace projector $\prod\nolimits_S \coloneqq \operatorname{Proj}(C)$ given by
\begin{equation}
\label{eq:projector}
    \prod\nolimits_{C(S)} = \frac{1}{|S|} \sum_{g\in S} g,
\end{equation}
is the projector onto the code space $C$.
\end{proposition}
\begin{proof}
Observe that $g\prod\nolimits_S=\prod\nolimits_S \forall g \in S$ since it acts on S by permuting its elements. Therefore $\prod\nolimits_S^2=\prod\nolimits_S$.
Hence $\prod\nolimits_S$ is a projector and has eigenvalues 0 or 1.
It is also a projector onto the code space, i.e., $\operatorname{Im}(\prod\nolimits_S) = C$.
Since  $g\prod\nolimits_S=\prod\nolimits_S \forall g \in S$ we have that $\operatorname{Im}(\prod\nolimits_S) \subseteq C$.
Furthermore, if $\ket{\phi}\in C$, then $\prod\nolimits_S \ket{\phi} = \ket{\phi}$ and therefore $C \subseteq \operatorname{Im}(\prod\nolimits_S) $
\end{proof}

For the above example, i.e., the pair $\ket{x}=\ket{10010}, \ket{y}=\ket{01011}$, we have that $\lX = X_1X_2X_5$.
However, there are more logical X operators,
which are all elements of $N\setminus S$ where $N$ is the normalizer of $S$.
All logical X operators are given by the set $\lX S = \{ \lX g \ | \ g\in S\}$.
Indeed, for any $\widehat{X} \in \lX S$ we have that $\widehat{X} \ket{x} = \lX g \ket{x} = \lX \ket{x} = \ket{y}$. %

Now we have all we need to show that $H_{x\leftrightarrow y}$ consists of all logical X operators in the following sense.
\begin{theorem}[Mixer as sum of logical X operators]
\label{theorem:MixerAsX}
Let $x\neq y \in B$, $S$ be the stabilizer of the code space $C=\Span{\ket{x},\ket{y}}$, and $\lX$ be as defined in Equation~\eqref{eq:Xhat}.
Then for $H_{x\leftrightarrow y} = \ket{x}\bra{y} + \ket{y}\bra{x}$ we have that 
\begin{equation}
\label{eq:HxyAsLX}
    H_{x\leftrightarrow y} = \frac{1}{|S|} \sum_{\widehat{X} \in \lX S} \widehat{X}.
\end{equation}
This means $H_{x\leftrightarrow y}$ is the sum of all logical X operators.
\end{theorem}
\begin{proof}
To show equality of the linear operators, it is enough to show equality for all computational basis states.
We start by observing that
\begin{equation}
\label{eq:XS}
    \frac{1}{|S|} \sum_{\widehat{X} \in \lX S} \widehat{X}
    =\frac{1}{|S|} \sum_{g \in S}\lX g
     = \lX \prod\nolimits_{C(S)},
\end{equation}
which means that the action of the right hand side in equation~\eqref{eq:HxyAsLX} is a projection into the code space followed by a logical X operation.
Both sides in equation~\eqref{eq:HxyAsLX} are zero for $\ket{\psi} \in (\C^2)^{\otimes n} \setminus C$.
Furthermore, we have that $\lX \prod\nolimits_S \ket{x} = \lX \ket{x} = \ket{y}$ and equivalently maps $\ket{y}$ to $\ket{x}$. Hence, also in this case both sides in equation~\eqref{eq:HxyAsLX} are equal, which proofs the assertion.
\end{proof}

\begin{corollary}[All non-vanishing Pauli strings commute]
\label{corollary:commuting}
Theorem 3 in~\replace{\cite{fuchs2022constraint}}{\cite{fuchs2021efficient}} now follows as a simple consequence by the construction via the stabilizer formalism.
    From equation~\eqref{eq:HxyAsLX} we can see that all non-vanishing Pauli strings of $H_{x\leftrightarrow y} $ commute, since they are in the normalizer of $S$, which in this case is equal to the centralizer of $S$.
\end{corollary}

\subsection{Stabilizer formalism for multiple pairs of states}
\label{sec:stabilizerPairsOfStates}

When there is more than one pair of states associated with the same logical operator, i.e. when $|E_{\lX}|=|V_{\lX}|/2>1$, we can in general no longer apply the stabilizer formalism directly for all states. %
One can still express the projector
$\prod\nolimits_{\Span{ V_{\lX} }}= \sum_{\ket{x}\in V_{\lX}} \ket{x}\bra{x}$
in Equation~\eqref{eq:Xprod} in the Pauli-basis, e.g. by using Equation~\eqref{eq:Paulidecomp}.
Since this is not necessarily optimal, we show how one can utilize the structure of the space subspace $V_{\lX}$ to construct mixers with as few Pauli-strings as possible. 

We begin with the following observation.
Let $\{ s_1, \cdots, s_l \}$ %
a minimal generating set of a stabilizer group $S=\langle s_1, \cdots, s_l \rangle$.
This defines the code space $C(S)=\Span{\ket{z} \ : \ s \ket{z} = \ket{z} \forall s \in S }$.
Let us now assume that there is a Pauli-string $E$ %
such that $E \ket{z} \notin C(S)$ for $\ket{z} \in C(S)$. In the error correction lingo $E$ would be interpreted as a detectable error.
Then we define $S^E = \langle s^E_1, \cdots, s^E_l \rangle$ with
\begin{equation}
    s^E_i = 
    \begin{cases}
        s_i,& \text{ if } [E, s_i]\add{=0},\\
        -s_i, & \text{ otherwise.}
    \end{cases}
\end{equation}
(Remember that Pauli strings either commute or anti-commute).
It follows that $S^E$ is a stabilizer group that defines the code space $C(S^E) =\Span{ E \ket{z} \ : \ \ket{z} \in C(S)}$. This is easy to see, since for any $\ket{z} \in C(S)$ we have that
\begin{equation}
    s^E_i E \ket{z} = \begin{Bmatrix}
        s_i E \ket{z},& \text{ if }    [s^E_i, E]=0, \\
        - s_i E \ket{z},& \text{ otherwise }
    \end{Bmatrix} =
    \begin{Bmatrix}
         E s_i \ket{z},& \text{ if }    [s^E_i, E]=0, \\
         - E (- s_i) \ket{z},& \text{ otherwise }
    \end{Bmatrix} =
    E \ket{z}.
\end{equation}
We observe that the terms in $S^E$ and $S$ only differ by sign.
Additionally, we know that $\exists s_i \in \{s_1, \cdots, s_l\}$ s.t. $\{s_i,E\} = 0$, i.e, $\exists s^E_i = - s_i$.
We define the index set 
\begin{equation}
    I^c =\{1\leq i \leq l \ | \ s_i=s^E_i\},
\end{equation}
i.e. the index set where the generators coincide, and $I^d = \{1\leq i \leq l\} \setminus I^c$, where they differ. Let
\begin{equation}
    I^{d,2} = \{ (i,j) : i\neq j, \ i,j\in I^d \}, 
\end{equation}
i.e., $s_i=-s^E_i \wedge s_j=-s^E_j$ for $(i,j)\in I^{d,2}$. 
These observations lead to the following
\begin{theorem}[Stabilizers for sum of orthogonal subspaces]
\label{theorem:stabilizersum}
Let $\{ s_1, \cdots, s_l \}$ 
be a minimal generating set of the stabilizer group $S=\langle s_1, \cdots, s_l \rangle$
defining the code subspace $C(S)$.
Let $E$ be a Pauli-string with $E \ket{z} \notin C(S)$ for all $\ket{z} \in C(S)$.

Then the stabilizer group
$S^{\langle E \rangle} = \langle s^{\langle E \rangle}_1, \cdots ,s^{\langle E \rangle}_{l-1}\rangle$ with the minimal generating set
\begin{equation}
    \{s^{\langle E \rangle}_1, \cdots, s^{\langle E \rangle}_{l-1}\} = \bigcup_{i\in I^c} \{s_i\}
    \bigcup_{(i,j)\in J} \{s_i s_j\},
\end{equation}
for any $J\subset I^{d,2}$ with $|J|=l-|I^c|-1$,
stabilizes the code space $C(S^{\langle E \rangle}) = C(S) \oplus C(S^E)$.

\end{theorem}
\begin{proof}
For all $i\in I^c$ we have that all states in $C(S)$ and $C(S^E)$ are a plus one eigenstate of $s_i$.
For $(i,j)\in I^{d,2}$ we have that $s_is_j = s^E_i s^E_j$ and therefore
all states in $C(S)$ and  $C(S^E)$ are a plus one eigenstates.
Notice that $\operatorname{dim}(C(S^{\langle E \rangle})) = \operatorname{dim}(C(S)) + \operatorname{dim}(C(S^E)) = 2^{n-l}+2^{n-l} = 2^{n-(l-1)}$, and therefore a minimal generating set has $l-1$ elements.
\end{proof}

The relationship $|S^{\langle E \rangle}| = |S|/2$ means that the number of Pauli strings to realize the projector onto $C(S^{\langle E \rangle})$ is halved as compared to $C(S)$.
This inspires the following algorithm to efficiently find a minimal generator set when the code space has an underlying group structure.
In the following we interpret $\{g_1, \cdots ,g_k\} \ket{\phi}$ as $\{g_1\ket{\phi}, \cdots ,g_k\ket{\phi} \}$.

{\centering
\begin{minipage}{1\linewidth}
\begin{algorithm}[H]
\caption{Minimal generator set of stabilizer of subspace $V=\Span{\langle E_1, \cdots, E_k \rangle \ket{z}}$}\label{alg:stabilizerC}
\textbf{Given:} $E_i\in \{I,X\}^n$, computational basis state $\ket{z}$

\textbf{Result:} Minimal generating set $G_k=\{ g_1, \cdots, g_{n-k}\}$ such that $C(\langle g_1, \cdots, g_{n-k}\rangle) = V$
\begin{algorithmic}
    \State $G_0= \{g_{0,1}, \cdots g_{0,n}\}$ with $g_{0,i}= (-1)^{z_i} Z_i$
    \For{$1\leq i \leq k$}
        \State Apply Theorem~\ref{theorem:stabilizersum} with $E=E_i$
               to get $G_i= \{g_{i,1}, \cdots g_{i,n-i}\}$
    \EndFor
\end{algorithmic}
\end{algorithm}
\end{minipage}
}

\begin{figure}
     \centering
     \begin{subfigure}[b]{0.42\textwidth}
     \centering
             \begin{tikzpicture}

    \node[] at (-4,.75) {$G_{\lX}$};
    
    \draw[fill=blue, opacity=.25] (-4,-.5) ellipse (.5cm and .95cm);
    \Vertex[style={color=blue, opacity=.0}, x=-3.75,y=-.5  ,label=$e_1$]{e1};
    \draw[fill=blue, opacity=.25] (-4,-2.5) ellipse (.5cm and .95cm);
    \Vertex[style={color=blue, opacity=.0}, x=-3.75,y=-2.5  ,label=$e_2$]{e2};
    \draw[fill=blue, opacity=.25] (-4,-5) ellipse (.5cm and .95cm);
    \Vertex[style={color=blue, opacity=.0}, x=-3.75,y=-5  ,label=$e_m$]{em};
    
    \Vertex[style={color=gray,opacity=.75},x=-4,y=0   ,label=$\ket{x_1}$]{e11};
    \Vertex[style={color=gray,opacity=.75},x=-4,y=-1,  label=$\ket{y_1}$]{e12};
    \Vertex[style={color=gray,opacity=.75},x=-4,y=-2  ,label=$\ket{x_2}$]{e21};
    \Vertex[style={color=gray,opacity=.75},x=-4,y=-3,  label=$\ket{y_2}$]{e22};
    \Vertex[style={color=white,opacity=0}, x=-4,y=-3.75  ,label=$\cdots$]{d};
    \Vertex[style={color=gray,opacity=.75},x=-4,y=-4.5  ,label=$\ket{x_i}$]{e31};
    \Vertex[style={color=gray,opacity=.75},x=-4,y=-5.5  ,label=$\ket{y_j}$]{e32};
    \Vertex[style={color=white,opacity=0}, x=-4,y=-6.25  ,label=$\cdots$]{d};

    \Edge[bend=-55,position=left, label=$\lX$](e11)(e12);
    \Edge[bend=-55,position=left, label=$\lX$](e21)(e22);
    \Edge[bend=-55,position=left, label=$\lX$](e31)(e32);
    
    \Edge[bend=+55,position=right, label=$E_2$](e1)(em);
    \Edge[bend=+55,position=right, label=$E_1$](e1)(e2);
    \Edge[bend=+55,position=right, label=$E_3$](e2)(em);
                
            \end{tikzpicture}
    \caption{
    }
     \end{subfigure}
     \hspace{1cm}
     \begin{subfigure}[b]{0.42\textwidth}
     \centering
         \begin{tikzpicture}
            \node[circle,draw=black, fill=white, inner sep=1pt, align=center] (x1) at (0,2-3.25) {$\ket{x_1}$};
            \node[circle,draw=black, fill=white, inner sep=1pt, align=center] (y1) at (3,2-3.25) {$\ket{y_1}$};
            \Edge[bend=35,label=$\lX$,position=above](x1)(y1)
            \node[circle,draw=black, fill=white, inner sep=1pt, align=center] (x2) at (0,0-3.25) {$\ket{x_2}$};
            \node[circle,draw=black, fill=white, inner sep=1pt, align=center] (y2) at (3,0-3.25) {$\ket{y_2}$};
            \Edge[bend=-35,label=$\lX$,position=above](x2)(y2)
            \Edge[bend=0,label=$E$,position=left](x1)(x2)
            \node[] at (5,1-3.25) {$\overset{[\lX, E]\add{=0}}{\Rightarrow} E\ket{y_1} = \ket{y_2}$};
            \node[] at (0,.75) {};
            \node[] at (0,-6.25) {};
        \end{tikzpicture}
    \caption{
    }
\end{subfigure}
\caption{
For each graph $G_{\lX}$ one can specify operators $\{E_i\}_i$ that map the states from one subspace spanned by the elements of $E_{\lX}$ to another.
(a) An illustration, where individual subspaces of $E_{\lX}$ are highlighted in light blue and $E_i$ map between them.
(b) An operator $E$ maps one subspace onto another, as long as $\lX$ and $E$ commute.
Indeed, we have that $E\ket{y_1} = E \lX \ket{x_1} = \lX E \ket{x_1} = \lX \ket{x_2} = \ket{y_2}$.
}
    \label{fig:Egroup}
\end{figure}
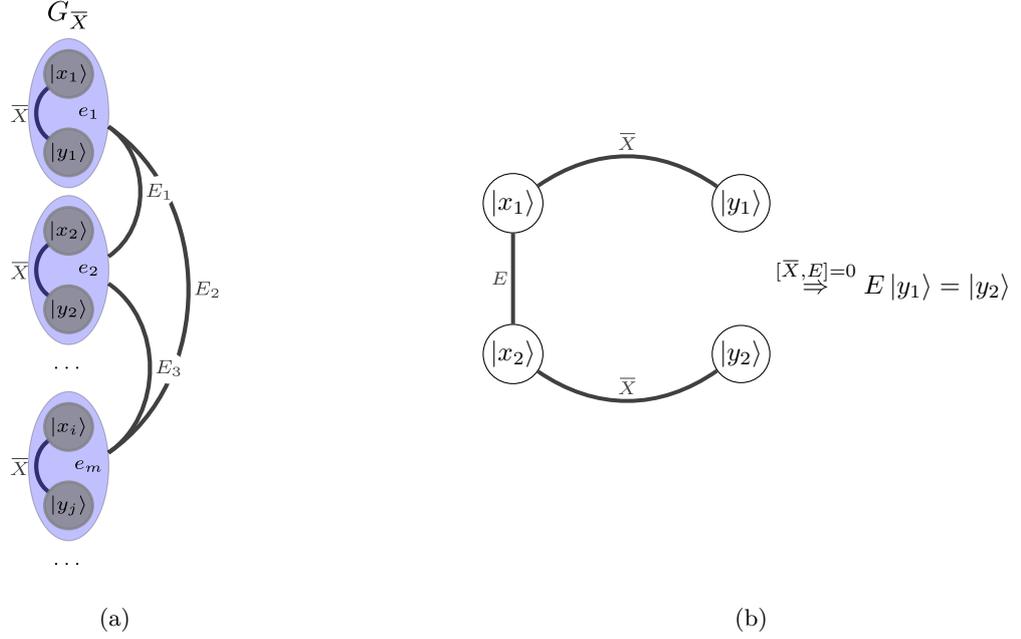

We illustrate Algorithm~\ref{alg:stabilizerC} for $V=\{
\ket{1011},\allowbreak \ket{1100}, \allowbreak
\ket{0111},\allowbreak \ket{0000}, \allowbreak
\ket{1110},\allowbreak \ket{1001}, \allowbreak
\ket{0010},\allowbreak \ket{0101}
\} \allowbreak =$ $ \allowbreak \langle X_2X_3X_4, \allowbreak X_1X_2, \allowbreak X_1X_4 \rangle \ket{1011}$.
In the first step we get that $G_0=\{-Z_1, Z_2, -Z_3, -Z_4\}$ stabilizes the state $\ket{1011}$.
Applying Theorem~\ref{theorem:stabilizersum} recursively
\begin{itemize}
    \item
    with $E=X_2X_3X_4$ gives $G_1=\{-Z_1, -Z_2Z_3, Z_3Z_4\}$,
    \item
    with $E=X_1X_2$ gives $G_2=\{Z_1Z_2Z_3, Z_3Z_4\}$,
    \item
    and with $E=X_1X_4$ gives $G_3=\{Z_1Z_2Z_4\}$.
\end{itemize}
It is easy to check that $C(\langle Z_1Z_2Z_4 \rangle) = \Span{V}$ and therefore the projector to that subspace is given by $\prod\nolimits_{\Span{V}} = \frac{1}{2}(IIII+ZZIZ)$.

\remove{
For every graph $G_{\lX}$ we can then take the set of operators $\{E_i\}_i$ between the subspaces spanned by the elements of $E_{\lX}$ .
In order to apply Algorithm~\ref{alg:stabilizerC}, we need to find (the largest) set of operators that form a group, which enables us to discover the projectors that have the lowest number of Pauli-strings.
}
\add{
\begin{definition}[Stabilizer subgraphs]
\label{def:stabilizersubgraphs}
For every graph $G_{\lX}$ we take the set of operators $\{E_i\}_i$ between the subspaces spanned by the elements of $E_{\lX}$. Note that $E_i \in \{I,X\}^n$ by construction.
Let $F$ be a subset of operators $F = \{ E_{i_1}, \cdots, E_{i_q} \} \subset \{E_i\}_i$.
We call $G_{\lX,F}$ a stabilizer subgraph of $G_{\lX}$ if for an $\ket{x}$ associated with a pair connected by $F$, we have that
\begin{equation}
    e \ket{x} \in V_{\lX}, \quad \forall e \in \langle E_{i_1}, \cdots, E_{i_q}\rangle.
\end{equation}
\end{definition}
An illustration of the set of operators $\{E_i\}_i$ is given in Figure~\ref{fig:Egroup}.
Algorithm~\ref{alg:stabilizerC} can be applied to all stabilizer subgraphs.
The larger the set $F$, the smaller the number of Pauli-strings required to form the projector to the subspace, since the stabilizer group has $2^{n}-|\langle F \rangle|$ elements.
}

\section{Restricting projectors to the feasible subspace}
\label{sec:restrictingprojectors}

We have seen in the previous sections mixers can be written as a sum of projection operators to certain subspaces followed by a logical X operator, i.e., $\lX_a \prod\nolimits_{\Span{V}}$ for some some subspace $V\subseteq V_{\lX}$.
In this section we discuss the possibility to reduce the cost of a mixer further by restricting the projector to the feasible subspace $\Span{B}$.
The essential properties of the projector $P_{V,B}\coloneqq \restr{\prod\nolimits_{\Span{V}}}{\Span{B}}$ that we need to fulfill are that
\begin{subequations}
\label{eq:PCond}
    \begin{align}
\label{eq:PV}
        P_{V,B}\ket{x} &= \ket{x},        & \ket{x} &\in V, \\
\label{eq:PBmVlX}
        P_{V,B}\ket{x} &= \mathbf{0},     & \ket{x} &\in B\setminus V_{\lX}.
    \end{align}
\end{subequations}
Since $\lX \ket{x} \in \Span{V_{\lX}}$ for all $\ket{x}\in V_{\lX}$, it does not matter whether $P_{V,B}$ projects the states in the (possibly empty) set $V_{\lX}\setminus V$ to zero or themselves as long as $P_{V,B}^2=P_{V,B}$.
We present two ways to achieve this goal.

\subsection{Linear algebra approach}
\label{sec:linalg}
From Section~\ref{sec:stabilizer} we saw that the projector can be written using the stabilizer formalism as $\prod\nolimits_{\Span{V}} = \frac{1}{|S|} \sum_{s\in S} s$, where $S$ is a stabilizer group with $C(S)=\Span{V}$.
Let $s_j$ be an enumeration of the elements in the group $S$.
One can now construct the matrix $A: \R^{|B\setminus V_{\lX}| \times |S|} \rightarrow \R$ with entries given by
\begin{equation}
\label{eq:A}
    A_{i,j} = \lambda^{s_j}_i, \quad s_j\in S, \ket{x_i}\in B\setminus V_{\lX},
\end{equation}
where $s_j\ket{x_i} = \lambda^{s_j}_i \ket{x_i}$.
In other words, $ \lambda^{s_j}_i\in \pm 1$ indicates whether $\ket{x_i}$ is in the $\pm 1$ eigenspace of $s_j\in S$.
This gives the possibility to reduce the number of Pauli-strings as follows.
\begin{theorem}[Projectors restricted to a feasible subspace]
\label{theorem:restrproj}
Let $S$ be a stabilizer group with $C(S)=\Span{V}$ and $A$ as defined Equation~\eqref{eq:A}. %
If $v\in\R^{|S|}$ such that $v\in \ker(A)$ and $\mathbf{1}^Tv = \sum_{i=1}^{|S|} v_i \neq 0$, %
then 
\begin{equation}
    P = \frac{1}{\sum_{i=1}^{|S|} v_i} \sum_{i=1}^{|S|} v_i s_i
\end{equation}
fulfills Conditions~\eqref{eq:PCond}.
\end{theorem}
\begin{proof}
First, we observe that $P$ is well-defined since $\sum_{i=1}^{|S|} v_i \neq 0$.
Let $\ket{x} \in \Span{V} = C(S)$ then by definition $s_i \ket{x} = \ket{x}$
and we have that
\begin{equation}
   P\ket{x} = \frac{1}{\sum_{i=1}^{|S|} v_i} 
   \sum_{i=1}^{|S|}  v_i s_i \ket{x} 
    = \frac{1}{\sum_{i=1}^{|S|} v_i}
    \sum_{i=1}^{|S|} v_i \ket{x}
    = \frac{1}{\sum_{i=1}^{|S|} v_i}
    \underbrace{\left(\sum_{i=1}^{|S|} v_i\right)}_{\neq 0} \ket{x} = \ket{x},
\end{equation}
so Condition~\eqref{eq:PV} is fulfilled.

If $\ket{x_i} \in B\setminus V_{\lX}$ we have that
\begin{equation}
   (\sum_{i=1}^{|S|} v_i) P\ket{x_i} = 
   \sum_{j=1}^{|S|}  v_j s_j \ket{x_i} =
   \sum_{j=1}^{|S|}  v_j A_{i,j} \ket{x_i} =
   (Av)_i \ket{x_i} = \mathbf{0},
\end{equation}
where the last equality holds because $v\in \ker(A)$.
This shows Condition~\eqref{eq:PBmVlX} is fulfilled.
\end{proof}

Since $P_{V,B}$ is a projector onto $V$ this implies that every row of $A$ has to have equally many plus one and minus one entries.
Therefore, a trivial solution of Theorem~\ref{theorem:restrproj} is $v=\mathbf{1}$.
However, in general there are other solutions, which can be found by bringing $A$ into \emph{reduced row echelon form}.

An example is given for $B=\{\ket{10010},\allowbreak \ket{01010}, \allowdisplaybreaks \ket{10011},\allowbreak \ket{11101},\allowbreak \ket{00110},\allowbreak \ket{01010}\}$ and 
$C_1=\operatorname{span}(\ket{10010},\allowbreak \ket{01110})$, where $A$ becomes the following matrix.
\begin{equation}
\label{eq:matrixA}
            \setlength{\arraycolsep}{3pt}
                \setcounter{MaxMatrixCols}{20}
                A=\begin{pNiceMatrix}[first-row,first-col]
                 &  \rotatebox[]{90}{$+ZIZZZ$}  &  \rotatebox[]{90}{$-ZZIIZ$}  &  \rotatebox[]{90}{$-ZIZIZ$}  &  \rotatebox[]{90}{$+ZZIZI$}  &  \rotatebox[]{90}{$-IZZZI$}  &  \rotatebox[]{90}{$-IZZZZ$}  &  \rotatebox[]{90}{$-ZIZII$}  &  \rotatebox[]{90}{$-IIIZI$}  &  \rotatebox[]{90}{$+ZIZZI$}  &  \rotatebox[]{90}{$+ZZIZZ$}  &  \rotatebox[]{90}{$+IIIIZ$}  &  \rotatebox[]{90}{$+IZZIZ$}  &  \rotatebox[]{90}{$-ZZIII$}  &  \rotatebox[]{90}{$-IIIZZ$}  &  \rotatebox[]{90}{$+IZZII$}  &  \rotatebox[]{90}{$+IIIII$} \\
            \ket{11101} & -1 & +1 & +1 & +1 & -1 & +1 & -1 & -1 & +1 & -1 & -1 & -1 & -1 & +1 & +1 & +1 \\
            \ket{01010} & -1 & +1 & -1 & +1 & -1 & -1 & -1 & +1 & -1 & +1 & +1 & -1 & +1 & +1 & -1 & +1 \\
            \ket{10011} & -1 & -1 & -1 & +1 & +1 & -1 & +1 & +1 & +1 & -1 & -1 & -1 & +1 & -1 & +1 & +1 \\
            \ket{00110} & +1 & -1 & +1 & -1 & -1 & -1 & +1 & +1 & +1 & -1 & +1 & -1 & -1 & +1 & -1 & +1 \\
                \end{pNiceMatrix}
\end{equation}
The rank of $A$ is 4, and therefore the kernel of $A$ has dimension 12, after the rank-nullity theorem. The subspace of the kernel where also $\mathbf{1}^Tv\neq0$ has dimension 10.
Two solutions with the lowest cost are given by $v=
(0, 1, 0, 0, 0, 0, 1, 0, 0, 0, 0, 0, 0, 0, 0, 0)
$ and $v=
(0, 0, 1, 0, 0, 0, 0, 0, 0, 0, 0, 0, 1,\allowbreak 0,\allowbreak 0, 0)
$
with
$H_{XXXII}=XXXII(-ZZIIZ -ZIZII)$
and
$H_{XXXII}=XXXII(-ZIZIZ -ZZIII)$, respectively.
Both reduced mixers bring down the cost from 96 to 10.

\subsection{Stabilizer formalism approach}
\label{sec:stabred}
An alternative approach can be derived using the stabilizer formalism.
Theorem~\ref{theorem:MixerAsX} and Equation~\eqref{eq:XS} gives us a way to construct valid mixers with fewer Pauli terms.
We start by observing that for any $\ket{\psi}\notin C$
there exists at least one $g_i$ in the minimal generating set of $S$ %
with $g_i\ket{\psi} = - \ket{\psi}$.
But since the eigenvalues of Pauli strings are $\pm 1$ and $\prod\nolimits_S \ket{\phi} = \mathbf{0}$ this means that half of the terms in $S$ must be $+1$ and the other half $-1$ eigenvectors of $\ket{\psi}$, otherwise they would not cancel out.
If we can find a subgroup $H<S$ that still has this characteristic for all $\ket{\psi} \in \Span{B}\setminus C$ then we need fewer Pauli-strings to realize $H_{x\leftrightarrow y}$ on the feasible subspace.
In general, we seek the following.
\begin{problem}[Optimal projector for subspace $V=\Span{\langle \lX_1, \cdots, \lX_k \rangle \ket{z}}$]
\label{prob:optimal}
Let $S$ be the stabilizer of the code space $V=\Span{\langle \lX_1, \cdots, \lX_k \rangle \ket{z}}$.
Then we want to find a smallest subgroup $H<S$ such that
\begin{equation}
    \restr{\frac{1}{|H|} \sum_{h \in H} h}{\Span{B}} = \restr{\prod\nolimits_s}{\Span{B}}. %
\end{equation}
Then by construction $P=\frac{1}{|H|} \sum_{h \in H} h$ fulfills Conditions~\eqref{eq:PCond}.
\end{problem}
The next question is how to find such a smallest subgroup.
Let $S=\langle g_1, \cdots, g_{n-k}\rangle$.
One can now construct the matrix $M: \R^{|B\setminus V_{\lX}|\times (n-k) } \rightarrow \R$ with entries given by
\begin{equation}
\label{eq:M}
    M_{i,j} = \lambda^{g_j}_i, \quad g_j\in S, \ket{x_i}\in B\setminus V_{\lX},
\end{equation}
where $g_j\ket{x_i} = \lambda^{g_j}_i \ket{x_i}$.
In other words, $M_{i,j}\in \pm 1$ indicates whether $\ket{x_i}$ is in the $\pm 1$ eigenspace of $g_j$.
Note, that in contrast to $A$ in Equation~\eqref{eq:matrixA}, $g_j$ is only taken from a minimal generating set of the group, not all elements of the group.
Since $ \ket{x_i}\in B\setminus V_{\lX}$ is outside the code space of $S=\langle g_1, \cdots, g_{n-k} \rangle$ is a generating set, this means that there has to be at least one $-1$ entry in each row. Otherwise $\ket{x_i}$ would belong to the code space.

One can create all minimal generating sets of $S$ by successively multiplying one element by another as long as they remain linearly independent, e.g, $\{g_1, \cdots , g_j, \cdots\allowbreak, g_l, \cdots, g_{n-k} \} \rightarrow \{ g_1, \cdots , g_j g_l, \cdots, g_l, \cdots, g_{n-k} \}$.
For the matrix $M$ this changing $g_j$ to $g_j g_l$ means multiplying column $j$ with $l$.
To solve Problem~\ref{prob:optimal} becomes then equivalent to finding the smallest subset of columns that have at least one $-1$ entry in each row, taking into account all possible multiplications of columns.

For the same example that was used in Section~\ref{sec:linalg}%
we have
\begin{equation}
\label{eq:matrixM}
            \setlength{\arraycolsep}{3pt}
                \setcounter{MaxMatrixCols}{10}
                M=\begin{pNiceMatrix}[first-row,first-col]
                 & \rotatebox[]{90}{$g_1=-IIIZI$} & \rotatebox[]{90}{$g_2=-IIIIZ$} & \rotatebox[]{90}{$g_3=-ZZIII$} & \rotatebox[]{90}{$g_4=+IZZII$} \\
            \ket{11101} & +1 & -1 & +1 & +1 \\
            \ket{01010} & -1 & -1 & -1 & +1 \\
            \ket{10011} & +1 & +1 & -1 & -1 \\
            \ket{00110} & +1 & +1 & +1 & -1 \\
                \end{pNiceMatrix}.
\end{equation}
In this case, multiplying columns $g_2$ and $g_4$ gives a column with all entries equal to $-1$. Therefore $H=\langle g_2g_4\rangle = \langle IZZIZ\rangle $ is a solution to Problem~\ref{prob:optimal} leading to
$H_{XXXII}=XXXII(IIIII + IZZIZ)$.
This reduces the cost from 96 to 10.

\subsection{Comparison}
The method in Section~\ref{sec:stabred} is based on a minimal generating set of the stabilizer group, whereas the one presented in Section~\ref{sec:linalg} uses all elements of the group.
Therefore, the matrix $M$ in Equation~\eqref{eq:matrixM} is exponentially larger in the number of group elements than the matrix $A$ in Equation~\eqref{eq:matrixA}.
However, the computational complexity for finding the projector restricted to the feasible subset with the lowest cost seems to be lower for the approach in Section~\ref{sec:linalg}, as it only involves bringing the matrix $A$ into reduced row echelon form.
An optimal solution to Problem~\ref{prob:optimal} seems to be computationally more complex (at least without further investigation).
There seems to be a trade-off between size and complexity.

Furthermore, the method in Section~\ref{sec:linalg} is more general, as it in principle works for any projector $P=\sum_i w_i P_i, P_i\in\{I,Z\}$ and the resulting solution does not need to be a group. This can also be seen in the example above, and it does not seem a-priorily obvious which of the two method gives the better solution.
An example of the reduction in cost is given in Table~\ref{tab:exampleH}.

\section{Numerical results of the algorithm}\label{sec:examples}

In this section we provide examples of the proposed approach.
The overall approach is \replace{described in Algorithm}{sketched in Figure}~\ref{fig:algoverall}, 
which allows to work with minimal generating sets until the last step, where the final mixer is generated.
\add{The overall complexity of the algorithm in Figure~\ref{fig:algoverall} is as follows:
\begin{itemize}
    \item Step (1) generates the family of valid mixers of at most $\mathcal{O}(|B|^2)$ graphs according to Definition~\ref{definition:MixerFamily} by going through all pairs of states in $B$. It should be noted that this is an upper bound: The larger the set $B$, the more likely it is for pairs to be in the same graph.
    \item Step (2) generates stabilizer subgraphs according to Definition~\ref{def:stabilizersubgraphs}. A necessary condition for this is that there are power of two edges, so we can upper bound this by $\sum_{i=0}^{\floor*{\operatorname{log_2}(|B|)}}\binom{|B|}{2^i} \leq \mathcal{O}(2^{|B|})$. This is a very rough upper bound.
    \item Step (3) consists of calculating the cost for all graphs, e.g., by solving Problem~\ref{prob:optimal}. A brute-force solver scales exponentially with $ B\setminus V_{\lX}$ so the runtime can be upper bounded by $\mathcal{O}(2^{|B|})$.
    \item Step (4) consists of finding the combination of stabilizer subgraphs that leads to a connected graph and has minimal cost, see Definition~\ref{def:optimalTrotterization}. A brute-force implementation looking through all possible combinations scales as $2^{|G|}$, where the number of graphs can be upper bounded by $|G| \leq |B|^2 2^{|B|}$.
\end{itemize}
We remark, that heuristics might be developed for better scaling that still lead to good solutions.
}

Figure~\ref{fig:statistics} shows how the proposed approach compares for randomly selected $B$.
We observe that the proposed method reduces \add{the cost of} the mixer that connects states linearly as in a chain.
This corresponds to a $T$ in Equation~\ref{eq:Hmdefinition} that has 1 in the first off-diagonal entry and 0 in in the rest of the entries.
The Figure shows that finding the optimal Trotterization reduces the cost (red line vs. blue line), particularly when $|B|$ is large. For the case when $|B|=2^n$, the proposed algorithm reproduces the standard $X$-mixer which has zero cost.
Figure~\ref{fig:X4} shows the optimal graph for the $X$-mixer, i.e., the unconstrained case.
Furthermore, the approach to restrict the projector to the feasible subspace, as described in Section~\ref{sec:restrictingprojectors}, can further reduce the required number of CNOT gates in a wide range of cases.

\begin{figure}
    \centering
    \input{pictures/concept}
    \caption{The steps of the overall algorithm. An optimal mixer is found by generating all possible stabilizer subgraphs and finding the connected graph composed by a subset of these subgraphs that has minimal cost.}
    \label{fig:algoverall}
\end{figure}

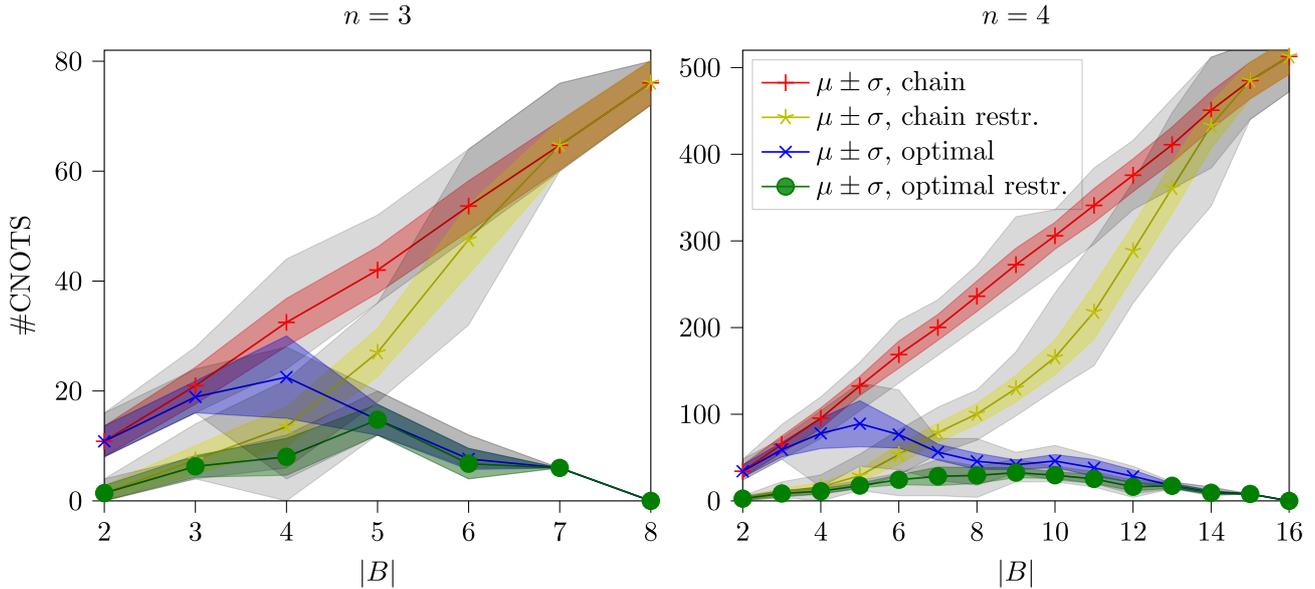
\begin{figure}
    \centering
    \begin{subfigure}[b]{0.45\textwidth}
    \input{pictures/statistics_cnots_n3_4.tex}
    \end{subfigure}
    \hspace{0.5cm}
    \begin{subfigure}[b]{0.45\textwidth}
    \input{pictures/statistics_cnots_n4_4.tex}
    \end{subfigure}
    \caption{Statistics of the number of controlled not gates required depending on the number of feasible states, i.e. $|B|$. For each $|B|$ 100 randomly selected states are drawn all computational basis states of $(C^2)^{\otimes n}$ and the corresponding restricted mixers are generated.
    The mean and one standard deviation of the cost is plotted in color.
    The gray shaded area is the range between the minimum and maximum obtained.
    We can see that the mixer corresponding to a linear chain of consecutive states is easily outperformed by finding an optimal Trotterization.
    We observe also that restricting the mixers to the feasible subspace reduces the cost.%
    }
    \label{fig:statistics}
\end{figure}

\begin{table}
\setlength{\tabcolsep}{3pt}
    \centering
            \begin{tabular}{rrrrr}
                \toprule
                 & standard & \add{ref.}\cite{fuchs2022constraint}& \replace{optimal}{LX-Mixer}\\
                \midrule
                $C_1=\{\ket{10010},\ket{01110}\}$ & 96  & 40 & 10 \\
                $C_2=\{\ket{10010},\ket{10011}\}$ & 64  & 24 & 4  \\
                $C_3=\{\ket{10010},\ket{11101}\}$ & 112 & 48 & 14 \\
                $C_4=\{\ket{10010},\ket{00110}\}$ & 80  & 32 & 10 \\
                $C_5=\{\ket{10010},\ket{01010}\}$ & 80  & 32 & 10 \\
                $C_6=\{\ket{01110},\ket{10011}\}$ & 112 & 48 & 14 \\
                $C_7=\{\ket{01110},\ket{11101}\}$ & 96  & 48 & 12 \\
                $C_8=\{\ket{01110},\ket{00110}\}$ & 64  & 24 & 4  \\
                $C_9=\{\ket{01110},\ket{01010}\}$ & 64  & 24 & 4  \\
                $C_{10}=\{\ket{10011},\ket{11101}\}$ & 96  & 40 & 10 \\
                $C_{11}=\{\ket{10011},\ket{00110}\}$ & 96  & 40 & 10 \\
                $C_{12}=\{\ket{10011},\ket{01010}\}$ & 96  & 40 & 10 \\
                $C_{13}=\{\ket{11101},\ket{00110}\}$ & 112 & 48 & 12  \\
                $C_{14}=\{\ket{11101},\ket{01010}\}$ & 112 & 48 & 12  \\
                $C_{15}=\{\ket{00110},\ket{01010}\}$ & 80  & 32  & 4  \\
                \bottomrule
            \end{tabular}
            \caption{
            Comparing the cost of the mixers for $B=\{\ket{10010},\allowbreak \ket{01010}, \allowdisplaybreaks \ket{10011},\allowbreak \ket{11101},\allowbreak \ket{00110},\allowbreak \ket{01010}\}$.
            When compared to the unrestricted mixer or the brute-force method in Fuchs et al.~\cite{fuchs2022constraint}, the cost can be considerably reduced with the approach described in Section~\ref{sec:restrictingprojectors}.
            }
    \label{tab:exampleH}
\end{table}

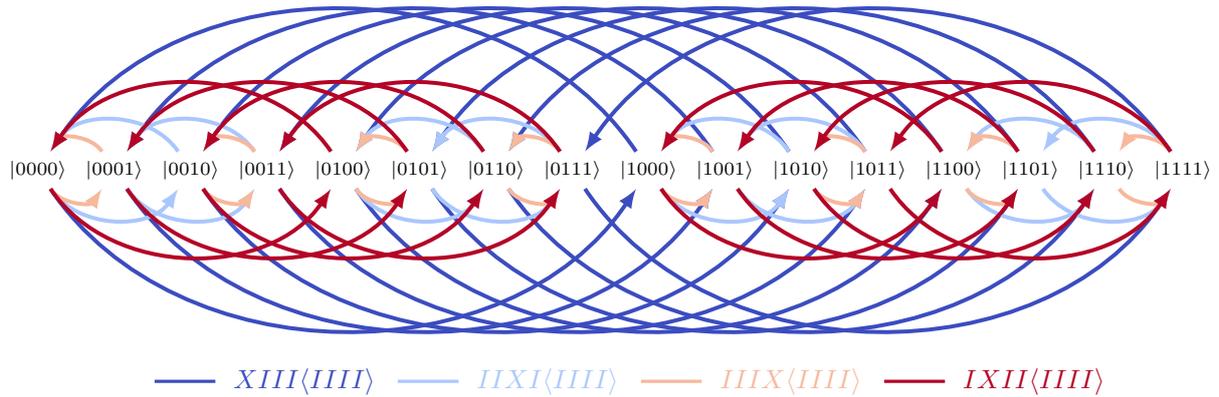
\begin{figure}
    \centering
    \input{pictures/X4}
    \caption{The proposed algorithm reproduces the standard "$X$"-mixer for the unconstrained case.}
    \label{fig:X4}
\end{figure}

The proposed method also reproduces the case which is sometimes referred to as "k-hot" or "XY"-mixer\replace{~\cite{fuchs2022constraint}
}{\cite{Wang2020}}.
The feasible subspace is defined by all computational basis states with exactly k ``1"s, i.e., $B=\{\ket{x}, \ x_j\in\{0,1\}^n, \ s.t. \sum{x_j}=k\}$. \add{As depicted in Figure \ref{fig:XY_64} for the 6 qubits 4-hot case each term in the mixer is of the form $XX+YY$ when restricted to the feasible subspace B.} 
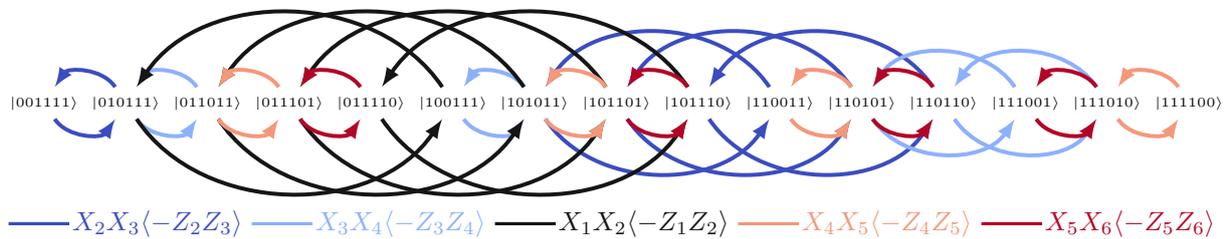
\begin{figure}
    \centering
    \input{pictures/XY_64}
    \caption{The proposed algorithm reproduces the standard "$XY$"-mixer for k-hot states.
    Here illustrated for $k=4$ with 6 qubits.
    }
    \label{fig:XY_64}
\end{figure}
An extension of this is shown in Figure~\ref{fig:Dicke514}, were $B$ consists of k-hot states for a range of different values k.
The approach shows that a low-cost mixer can be constructed as a combination of XY- and X-mixers.

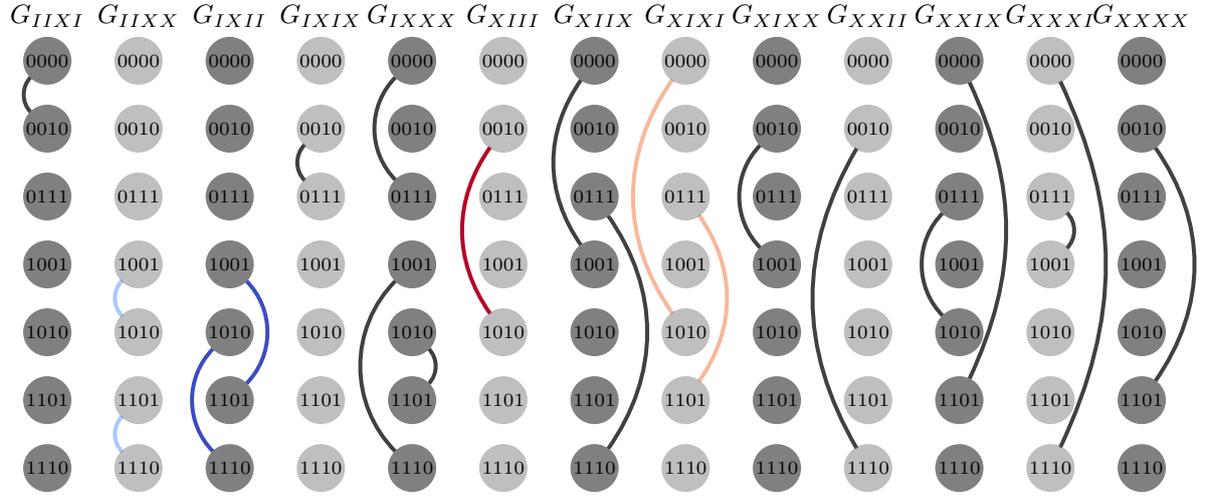
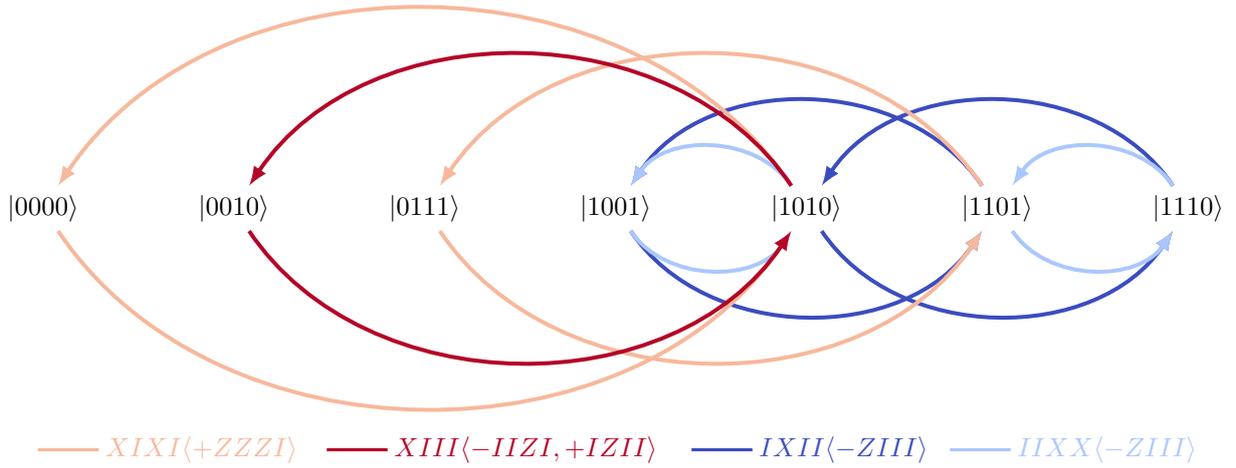
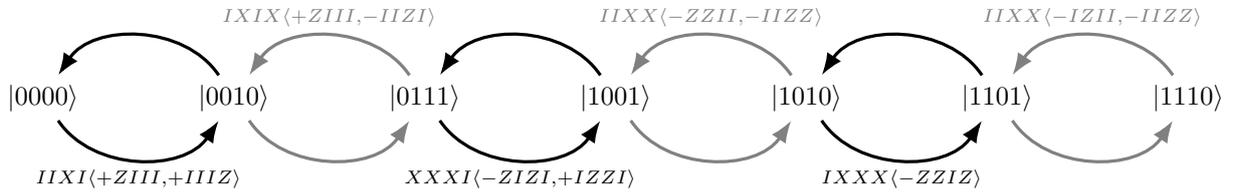
\begin{figure}
    \centering

     \begin{subfigure}[b]{0.98\textwidth}
     \centering
    \input{pictures/ex2_family}
    \caption{}
    \end{subfigure}
    
     \begin{subfigure}[b]{0.98\textwidth}
     \centering
    \input{pictures/ex2}
    \caption{}
    \end{subfigure}

\vspace{-1\baselineskip}
     \begin{subfigure}[b]{0.98\textwidth}
     \centering
    \input{pictures/ex2_chain}
    \caption{}
    \end{subfigure}

    \caption{(a) Example of family of graphs with commuting pairs for $B=\{\ket{1010}, \ket{0111}, \ket{1110}, \ket{1001}, \ket{0010}, \allowbreak \ket{0000}, \ket{1101}\}$.
    (b) The optimal Trotterization as proposed in this work has cost 64/22 (with/without restriction to feasible subspace).
    (c) The linear "chain" mixer has cost 200/98 (with/without restriction to feasible subspace).}
    \label{fig:comm_pairs_2}
\end{figure}

Finally, we illustrate the family of graphs
for the case when  $B=\{\ket{1010}, \ket{0111}, \allowbreak  \ket{1110},  \allowbreak \ket{1001}, \allowbreak  \ket{0010}, \allowbreak \ket{0000}, \ket{1101}\}$ in Figure~\ref{fig:comm_pairs_2}.
Comparing the "chain" mixer to the optimal solutions results in a roughly 70\% fewer CNOT gates.
Restricting the mixers to the feasible subspace further reduces the number of CNOT gates by 75\% and 51\%, respectively.

\section{Applications and practical considerations}
The numerical results discussed previously demonstrate a promising reduction in cost, specifically in terms of CNOTs, through the implementation of a general algorithm for optimal mixer circuits.
Despite this advancement the algorithm will in general scale prohibitively with $n$ when applied to general unstructured set $B$. %
The reason for this is that the family of graphs scale exponentially with respect to the number of vertices %
that must be evaluated.
This restricts the feasibility of this approach for unstructured problems, limiting its application to small-scale cases.
Nevertheless, the introduced formalism holds promise for discovering (near-)optimal mixers in terms of CNOTs when a structured feasible subset is provided by the problem at hand.
In this section, we show one way the proposed method scales favourably with the number of qubits when the bitstrings of the feasible subspace admit a Cartesian product structure.
Furthermore, we show a relaxation of the k-hot case to the multi-k-hot case.
Finally, we show a numerical simulation of a small case combining the two approaches \add{employing the LX-QAOA ansatz on a 10 variables optimization problem}.

\subsection{Feasible sets with Cartesian product structure}
\label{sec:cartesianprod}
When the feasible bitstrings admit a Cartesian product structure, the feasible subspace becomes a tensor product of the form
\begin{equation}
    B = B^1\otimes \cdots \otimes B^L.
\end{equation}
\begin{figure}
    \centering
    \input{pictures/cartesian_composition}
    \caption{
    An example of the box product construction for $B=\{ \ket{s_1 s_2}, s_1\in\{10,01\}, s_2\in \{100, 010, 001\}\}$.
    For this case $a_1=100, a_2=010, a_3=001$, $b_1=10, b_2=01$, with optimal mixers $H=H_1+H_2$, where $H_1=X_1X_2\langle -Z_1Z_2\rangle + X_2X_3\langle -Z_2 Z_3 \rangle $ and $H_2=X_1X_2$.
    }
    \label{fig:cartesian}
\end{figure}
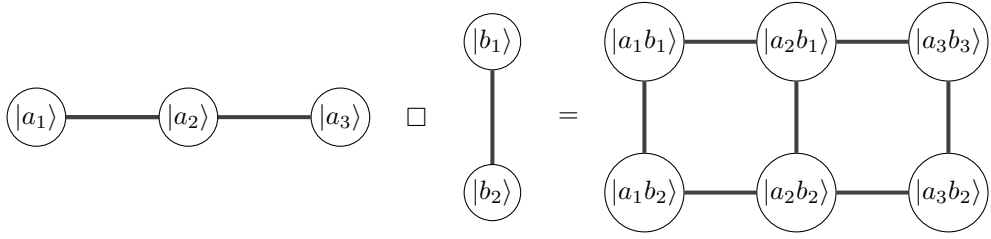
For this case it is possible to find an optimal mixer by employing the proposed algorithm independently on each subspace $B_i$.
The algorithm remains scalable if there are a fixed number of different cases $B_i$ when $L$ is increased.
Given the independence of the constraints, we can then define the following mixer Hamiltonian
\begin{equation}
    H_{M}=\sum_{i=1}^{L} H_i,
\end{equation}
where $H_i$ acts non-trivially on the subspace according to $B_i$.
From algebraic graph theory, it is known that the Kronecker sum of adjacency matrices corresponds to what is commonly referred to as the Cartesian product of graphs, see Figure~\ref{fig:cartesian}.

{\centering
 \begin{minipage}{1\textwidth}
\begin{definition}[Cartesian product of graphs]
\label{def:Cartesianproduct}
Let G(V,E) and H(V',E') be two graphs, the Cartesian product $G\Box H$, usually called box product, is a graph such that:
\begin{itemize}
    \item The vertex set of $G(V,E)\Box H(V',E')$ is composed by the Cartesian product of $VxV'$
    \item Two vertex (u,u') and (v,v') are adjacent in $G\Box H$ if
    \begin{itemize}
        \item u = v and u' and v' are adjacent in H
        \item u'= v' and u and v are adjacent in G 
    \end{itemize}
\end{itemize}
\end{definition}
 \end{minipage}
}
By construction, each distinct $H_{i}$ defines a valid \add{LX-}mixer for $B_i$ represented by a connected adjacency matrix. Therefore we need to verify that the adjacency matrix of the graph associated to the new Hamiltonian $H_{M}$ is connected too or more specifically that the Cartesian product of connected graphs is connected.

\begin{theorem}[Cartesian product Connectivity]
\label{theorem:connCartesian}
    Let G(V,E) and H(V',E') be connected graphs then $G\Box H$ is connected.
\end{theorem}
\begin{proof}
Consider two connected graphs G and H. Let $V(G) = \left\{v_1, v_2, ..., v_n\right\}$ and $V(H) = \left\{u_1, u_2, ..., u_m\right\}$ be the vertex sets of G and H, respectively.
For any vertex $v_i \in V(G)$, consider the set of vertices $V_{v_i} = \left\{(v_i, u_j) | u_j \in V(H)\right\}$ in the Cartesian product graph $G\Box H$.
Since H is connected, for each $v_i$, there exists a path in H connecting any two vertices $u_a$ and $u_b$. Therefore, for any pair of vertices $(v_i, u_a)$ and $(v_i, u_b)$ in $V_{v_i}$, there exists a path in $G\Box H$ connecting them, following the Cartesian product definition.
The same argument can be used for any vertex $u_i \in V(H)$, consider the set of vertices $V_{u_i} = \left\{(v_j, u_i) | v_j \in V(G)\right\}$ in the Cartesian product graph $G\Box H$ and using the Cartesian product definition.

Now consider any pair of vertices $(v_i, u_a)$ and $(v_k, u_b)$ in $G\Box H$.
We can construct a path between $(v_i, u_a)$ and $(v_k, u_b)$ as follows:
\begin{enumerate}
    \item Move from $(v_i, u_a)$ to $(v_i, u_b)$ along a path in $V_{v_i}$, which exists as shown earlier since G is connected. %
    \item Move from $(v_i, u_b)$ to $(v_k, u_b)$ along the path in $V_{u_a}$, which exists as shown earlier since H is connected. %
\end{enumerate}
Therefore, we have constructed a path in $G\Box H$ connecting $(v_i, u_a)$ and $(v_k, u_b)$ for any pair of vertices, proving that $G\Box H$ is connected.
\end{proof}

An example of the construction is given in Figure~\ref{fig:cartesian}.
The commutativity property of the Cartesian product correspond to the fact that the terms in the Hamiltonian commute with each other since they are applied to different subsets of qubits. This means that no Trotterization is needed for each term of the Hamiltonian and the total cost $Cost(H_{M})=\sum_{i}Cost(H_{i})$ is the sum of the optimal costs for each subsets of bits.
In conclusion, if the constraints for each independent subset of the bit string are small and unstructured, we can still use the logical X algorithm to find the optimal mixers for each set of qubits and then tensorize them.

\subsection{A mixer for the multi-k-hot case}
\label{sec:multikhot}
 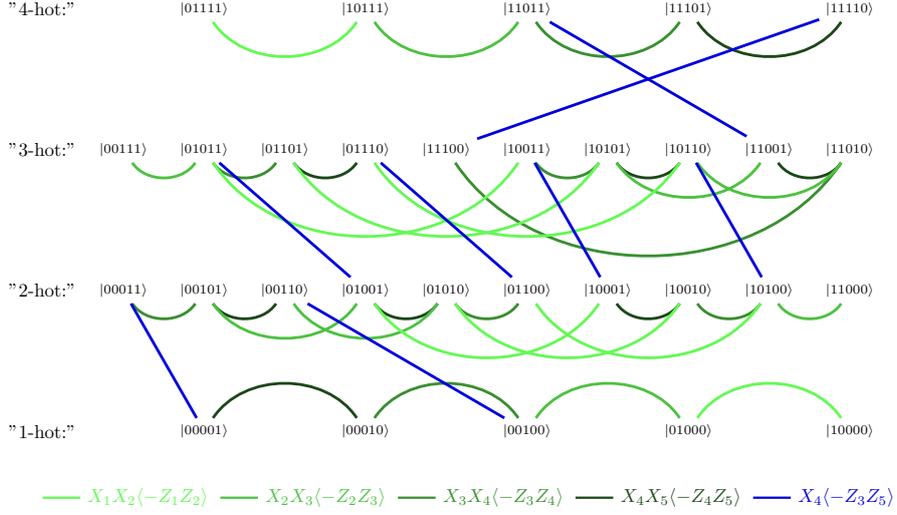
\begin{figure}%
    \centering
    \begin{adjustbox}{width=0.7\linewidth}
    \input{pictures/Dicke_5_1_4}
    \end{adjustbox}
    \caption{Example graph for 
   $ 
    B=\left\{\ket*{z}, \ z\in\{0,1\}^5, 1 \leq \sum_j z_j \leq 4 \right\}
    $, i.e. $B$ consist of all computational basis states except $\ket{00000}$ and $\ket{11111}$.
    One of the optimal mixers is shown in the graph.
    It consist of the XY-mixer which mixes all states with the same number of ones (green arrows), plus a restricted $X$-mixer, which connects sates with different number of ones (blue arrows).
    }
    \label{fig:Dicke514}
\end{figure}
We consider scenarios where the number of ones is not constrained to a single value, but to for instance to a range $0\leq k_1\leq k \leq k_2\leq n$. We define the feasible set $B_{k_{1},k_{2}}$ as
\begin{equation}
    B_{k_{1},k_{2}}=\left \{\ket{z},z\in\left \{ 0,1 \right \}^{n}, k_{1}\leq\sum_{i}^{n}z_{i}\leq k_{2} \right \}.
    \label{eq:k1k2set}
\end{equation}
We already know that an appropriate XY-mixer provides an undirected and connected graph that connects all the states in a fixed $B_{k}$.
This means  we need to look for a new set of logical Xs that provides transitions between different $B_{k}$ sets. Although identifying an optimal solution becomes more challenging in this context, as will be clear later, with the stabilizer formalism we can still construct an efficient mixer.

For this construction, we utilize a logical X operator comprised solely of a single Pauli $X_{i}$ applied to a specific qubit i. By design the operator flips the i-th qubit providing transitions between $B_{k}$ and $B_{k\pm1}$.
To ensure the condition that defines $B_{k_{1},k_{2}}$ we still need to avoid transition from $B_{k_{1}}$ to $B_{k_{1}-1}$ and $B_{k_{2}}$ to $B_{k_{2}+1}$.
Let's start by constructing the sets of states in $\ket{v}\in B_{k_{1}}\cup B_{k_{2}}$ such that $X\ket{v}\notin B_{k_{1},k_{2}} $.
\begin{equation}
    V^{X_{i}}_{k_{1},k_{2}}=\left\{\ket{v}\in B_{k_{1}}\cup B_{k_{2}}\quad s.t\quad X_{i}\ket{v}\notin B_{k_{1}+1}\cup B_{k_{2}-1} \right\}
\end{equation}
Next, we construct the projector restricted to the feasibel subspace, i.e.,
\begin{equation}
    \Pi=\mathbb{I}-\restr{\Pi_{V^{X_{i}}_{k_{1},k_{2}}}}{span B_{k_{1},k_{2}}}.
\end{equation}

By implementing these strategy, we constructed an efficient \add{LX-}mixer capable of providing transitions between different $B_{k}$ sets while respecting the specified constraints.
In this case many different projectors can be constructed and so to find the optimal mixer we are suppose to look for the projector that minimizes the CNOT cost function. %
Figure~\ref{fig:Dicke514} shows the construction for the case of $B_{1,4}$.

\subsection{Simulation example for a constrained MAXCUT instance}
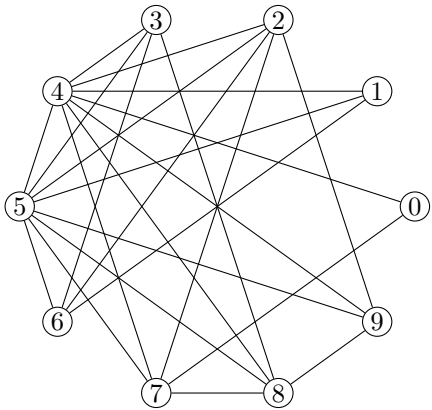
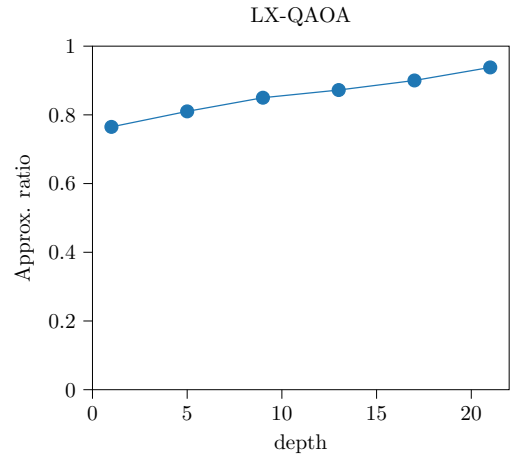
\begin{figure}
     \centering
     \begin{subfigure}[b]{0.42\textwidth}
     \centering
     \input{pictures/my_figure}
        \caption{10 vertex Barabási–Albert graph instance. The weights for this specific case are distributed uniformly between (0,1) and are available at \url{https://github.com/OpenQuantumComputing}.}
     \end{subfigure}
     \hspace{1cm}
     \begin{subfigure}[b]{0.42\textwidth}
     \centering
     \input{pictures/approx_ratio}
     \caption{Approximation ratio \add{$\alpha$} \remove{results} using the LX-QAOA ansatz with the COBYLA optimizer\add{, with $\alpha=1$ indicating that the QAOA achieves the optimal solution}.}
    \end{subfigure}
    \caption{
The LX-QAOA ansatz has been employed up to depth p=21 using the gradient free COBYLA optimizer for the optimization of the cost function \eqref{eq:model}. The optimal solution has been computed brute-force and compared to the QAOA results using the approximation ratio defined as $\alpha=\frac{C}{C_\text{opt}}$\add{,
where $C$ represents the expectation value for QAOA and $C_\text{opt}$ denotes the cost of the optimal solution.}
    }
    \label{fig:maxcut}
\end{figure}
Finally, we provide a numerical example for the specific case of a constrained MAXCUT as a proof-of-principle for \add{LX-}QAOA.
The problem has constraints with the Cartesian product structure of Section~\ref{sec:cartesianprod}, where each subspace is a multi-k-hot state, more specifically $B = B_{0,1} \otimes B_{0,1}$.
The optimization problem is formulated as follows:
\begin{equation}
        \underset{\{x_{i,j}\}_{i,j}}{\operatorname{max}} \quad  \sum_{i,k=1}^2 \sum_{j,l=1}^{5} x_{i,j}(1-x_{k,l}) C_{ijkl}, \quad x_{i,j} \in \{0,1\}, \quad \sum_{j=1}^{5}  x_{i,j} \leq 1, \forall i.
    \label{eq:model}
\end{equation}
The quadratic cost function is encoded in the weights of the graph, as illustrated in Figure~\ref{fig:maxcut}~(a).
For this constraint MAXCUT problem, we introduce a 5 vertices constraint for each of the 2 subsets, ensuring that the state for each set of vertices is either $\ket{0}$ or a 1-hot state.
Figure~\ref{fig:maxcut} shows the simulation results, indicating a good convergence of the approximation.

Starting in the uniform superposition of all feasible states, we run the QAOA algorithm with the \add{LX-}mixer given by $H = I^{5}\otimes H_{0,1} + H_{0,1} \otimes I^{5} $, with
\begin{equation}
    H_{0,1} = X_1X_2\langle -Z_1Z_2\rangle +
        X_2X_3\langle -Z_2Z_3\rangle +
        X_3X_4\langle -Z_3Z_4\rangle +
        X_4X_5\langle -Z_4Z_5\rangle +
        X_1\langle Z_2Z_3Z_4Z_5\rangle.
\end{equation}
Simulation results show that the approximation ratio increases with depth, see Figure~\ref{fig:maxcut}~(b).
By design the approach only gives feasible solutions.
The approach can be easily scaled to larger instances by design.

\section{Availability of Data and Code}
All data and the python/jupyter notebook source code for reproducing the results obtained in this article are available at  \url{https://github.com/OpenQuantumComputing} \add{in the repositories LogicalXMixer, and QAOA}.

\section{Author contributions}
Franz G. Fuchs formulated the concept, developed the methodology, wrote the software, made the formal analysis and investigation, wrote the article, made the visualizations.
Ruben Pariente Bassa, contributed with the section on applications, including concepts, analysis, software simulations, writing and visualizations.

\section{Acknowledgment}
We would like to thank for funding of the work by the Research Council of Norway through project number 332023 (80\%) and 324944 (20\%).
The authors wish to acknowledge the helpful discussions with Aleksandra Glesaaen, Giorgio Sartor, Kjetil Olsen Lye, Alexander Müller-Hermes, Alexander Stasik, Carlo Mannino, and Erlend Syljuåsen.
\add{We would also like to thank the reviewers for taking the time and effort necessary to review the manuscript.}

\section{Conclusion and Future Work}
We have presented a useful connection between the construction of mixers that act only on subspaces of the full Hilbert space and the stabilizer formalism.
The proposed method provides a general approach to construct resource efficient mixers for a given feasible subspace.
The numerical results show a large reduction of the number of CNOT gates required.
One limitation of this work is that the feasible subspace has to be known beforehand, so the method will for instance not be useful if it is NP-hard to even find a feasible solution.
\add{The proposed method is only applicable to problem domains that have an associated indexing function, which efficiently assigns a unique integer index to each feasible solution.}
Another limitation is that we do not expect our approach to scale to large dimensions of the Hilbert space.
\add{For an unstructured feasible set $B$, the presented algorithm has a worst case scaling of $\mathcal{O}(2^{2^{|B|}})$.}
However, the method might still be useful in the following case: When the size of the problem increases, the structure of the subspace has to be such that there is a systematic way to describe the resulting constrained mixers for larger instances.
\remove{This can be similar to Section~\ref{sec:cartesianprod}.}
An example of this would be the "XY"-mixer, which is straight forward to generalize to any number of qubits\add{, or along the lines of Section~\ref{sec:cartesianprod}.}
In that sense, our approach can be seen as a discovery method for efficient constrained mixers for structured problems.
Future work includes deriving efficient methods to prepare initial states within the feasible subspace for QAOA, as well as a mathematical analysis of the complexity of the graph problems for optimal Trotterization.
Even though it might turn out to be NP-hard to find an optimal Trotterization, efficient heuristics \add{to lower the scaling of the algorithm} can be applied to reduce the cost of the circuits of the mixers considerably, without necessarily finding the optimal solution.

\sloppy
\emergencystretch=1em
\printbibliography
\end{document}

%% file: pictures/adjacency.tex
 \begin{tikzpicture}
    \node at (-8.75,0) {
    $
        T=\bordermatrix{
                      & \ket{z_1}  \! \! \! \! \! \! & \ket{z_2}    \! \! \! \! \! \!   & \cdots  \! \! \! \! \! \! & \ket{z_i}    \! \! \! \! \! \!   & \ket{z_j}    \! \! \! \! \! \!   & \cdots  \! \! \! \! \! \! & \ket{z_n}    \! \! \! \! \! \!   \cr
        \ket{z_1}     &            & \color{red}1          &  &     1          &               &        &       \color{green}1     \cr
        \ket{z_2}     & \color{red}1          &            &  &                &               &        &            \cr
        \vdots &  &  &  &   &  &  &  \cr
        \ket{z_i}     &      1      &            &  &               &        \color{blue}1      &        &            \cr
        \ket{z_j}     &            &            &  &        \color{blue}1      &               &        &            \cr
        \vdots &  &  &  &   &  &  &  \cr
        \ket{z_n}     & \color{green}1           &            &  &                &               &        &            \cr
                }
    $};
    \node [arrowstyle=2] at (-4,0) {\phantom{matrix}};
            \def \n {5}
            \def \radius {1.5}
              \foreach\s in{1,...,\n}{
                  \node[circle,draw=black, fill=white, inner sep=1pt] (z\s) at (-360/\n*\s:-\radius) {$\ket{z_\ifnum\s=1\relax 1\fi\ifnum\s=2\relax 2\fi\ifnum\s=3\relax i\fi\ifnum\s=4\relax j\fi\ifnum\s=\n\relax n\fi}$};
              };
    \Edge[color=red,position=right](z1)(z2);
    \Edge[color=blue,position=right](z3)(z4);
    \Edge[color=green,position=right](z1)(z5);
    \Edge[color=black,position=right](z1)(z3);
\end{tikzpicture}

%% file: pictures/graph_minimazation_problem.tex
\begin{tikzpicture}[scale=.58]%

\newcommand{\drawNodes}[3]{
    \node[] at (#2,7.1) {$G_{#3}$};
    \Vertex[style={color=gray},x=#2,y=0  ,label=0011]{0011};
    \Vertex[style={color=gray},x=#2,y=1.5,label=0100]{0100};
    \Vertex[style={color=gray},x=#2,y=3  ,label=1001]{1001};
    \Vertex[style={color=gray},x=#2,y=4.5,label=1100]{1100};
    \Vertex[style={color=gray},x=#2,y=6  ,label=1110]{1110};
}

\newcommand{\drawNodesGray}[3]{
    \node[] at (#2,7.1) {$G_{#3}$};
    \Vertex[style={color=gray},opacity=.25,x=#2,y=0  ,label=0011]{0011};
    \Vertex[style={color=gray},opacity=.25,x=#2,y=1.5,label=0100]{0100};
    \Vertex[style={color=gray},opacity=.25,x=#2,y=3  ,label=1001]{1001};
    \Vertex[style={color=gray},opacity=.25,x=#2,y=4.5,label=1100]{1100};
    \Vertex[style={color=gray},opacity=.25,x=#2,y=6  ,label=1110]{1110};
}

\def\dist{3}

\drawNodes{0}{0*\dist}{IIXI}
\Edge[bend=-45,position=left](1110)(1100)

\drawNodes{1}{1*\dist}{IXIX}
\Edge[bend=+45,position=left](1001)(1100)

\drawNodes{2}{2*\dist}{IXXX}
\Edge[bend=+45,position=left](1001)(1110)
\Edge[bend=+45,position=left](0100)(0011)

\drawNodes{3}{3*\dist}{XIII}
\Edge[bend=+45,position=left](0100)(1100)

\drawNodes{4}{4*\dist}{XIXI}
\Edge[bend=+45,position=left](1001)(0011)
\Edge[bend=-35,position=left](1110)(0100)

\drawNodes{5}{5*\dist}{XXIX}
\Edge[bend=+45,position=left](1001)(0100)
\Edge[bend=-29,position=left](1110)(0011)

\drawNodes{6}{6*\dist}{XXXX}
\Edge[bend=+35,position=left](1100)(0011)

\drawNodes{7}{7*\dist}{IIXI}
\Edge[bend=+45,position=left](1110)(1100)

\end{tikzpicture}

%% file: pictures/concept.tex
\tikzstyle{blocktall} = [rectangle, draw, fill=blue!20, minimum width=2cm, text width=4.2cm, text centered, rounded corners, minimum height=5.5cm,inner sep=2pt, outer sep=1pt]
\tikzstyle{block} = [rectangle, draw, fill=blue!20, minimum width=2cm, text width=4.2cm, text centered, rounded corners, minimum height=.75cm,inner sep=2pt, outer sep=1pt]
\tikzstyle{circ} = [circle,fill=lightgray,draw,inner sep=0pt, outer sep=0pt,minimum size=4pt]
\tikzstyle{line} = [draw, -latex']

\newcommand{\drawNodes}[3]{
    \node[] at (#2,8.75) {$G_{#3}$};
    \node[circ] at (#2,7.5) (z1) {$z_1$};
    \node[circ] at (#2,6.5) (z2) {$z_2$};
    \node[circ] at (#2,5.5) (z3) {$z_3$};
    \node[circ] at (#2,4.5) (z4) {$z_4$};
    \node[circ] at (#2,3.5) (z5) {$z_5$};
    \node[circ] at (#2,2.5) (z6) {$z_6$};
    \node[] at (#2,1.75) (zd) {$\vdots$};
    \node[circ] at (#2,.5) (zJ) {$z_J$};
}

\newsavebox\graphbox
\begin{lrbox}{\graphbox}
        \begin{tikzpicture}[scale=.48]
            \def\dist{2}
            \drawNodes{0}{0*\dist}{\lX_1}
            \draw [-] (z2) to [bend right=45]  node[midway,left] {$E_1$} (z4);
            \draw [-] (z6) to [bend right=45]  node[midway,left] {$E_k$}  (0*\dist-.25, 1);
            
            \node[] at (1*\dist,8.75) {$\cdots$};
            
            \drawNodes{0}{2*\dist}{\lX_m}
            \draw [-] (z1) to [bend right=45]  (z3);
            \draw [-] (z5) to [bend right=45]  (zJ);
        \end{tikzpicture}
\end{lrbox}

\newcommand{\drawsubNodesA}[3]{
    \def\xm{9}
    \def\xd{0.5}
    \def\sd{2}
    
    \node[] at (#2-\sd,0-.25) {$\vdots$};
    \foreach \x in {1,...,\xm}
        \node[circ] at (#2-\sd,\x*\xd-.25) (zm\x) {};
        
    \node[] at (#2+\sd,0-.25) {$\vdots$};
    \foreach \x in {1,...,\xm}
        \node[circ] at (#2+\sd,\x*\xd-.25) (zp\x) {};

    \node[] at (#2,\xm*\xd+2.5) (G) {$G_{#3}$};
    \node[] at (#2-\sd,\xm*\xd+.5) (Gm) {$G_{#3,1}$};
    \node[] at (#2+\sd,\xm*\xd+.5) (Gp) {$G_{#3,p_i}$};
    \node[] at (#2,\xm*\xd+.5) {$\cdots$};
    \draw [->] (G) to  (Gm);
    \draw [->] (G) to  (Gp);
}

\newsavebox\subgraphbox
\begin{lrbox}{\subgraphbox}
        \begin{tikzpicture}[scale=.48, anchor=center]
            \def\dist{2}
            \drawsubNodesA{0}{0}{\lX_i}
            \draw [-] (zm9) to [bend right=90]  (zm8);
            \draw [-] (zp9) to [bend right=90]  (zp8);
            \draw [-] (zp5) to [bend right=90]  (zp1);
            \draw [-] (zp6) to [bend right=90]  (zp3);
            \draw [-] (zp4) to [bend right=90]  (zp2);
        \end{tikzpicture}
\end{lrbox}

\begin{tikzpicture}%
    \node [block] (init) {\textbf{Input:}\\ $B = \{ \ket{z_1}, \cdots, \ket{z_J}\}$};
    \node [blocktall, below=.75cm of init] (graphs) {
        \begin{minipage}{\linewidth}
        \textbf{(1)} Family of valid mixers:
        \end{minipage}
        \\ \usebox\graphbox };
    \node [blocktall, right=.75cm of graphs] (subgraphs) {
        \begin{minipage}{\linewidth}
        \textbf{(2)} 
        Find stabilizer sub-
        \phantom{\textbf{(2)}}
        graphs for all $i$:
        \end{minipage}
        \\ \usebox\subgraphbox };
    \node [block, right=.75cm of subgraphs, yshift=-0cm] (cost) {
        \begin{minipage}{\linewidth}
        \textbf{(3)} Calculate cost for all
        \phantom{\textbf{(3)}}
        stabilizer sub-graphs \\
        \phantom{\textbf{(3)}}
        $G_{\lX_i,j}$
        \end{minipage}
        };
    \node [block] (out) at (init -| cost){ \textbf{Output:}\\Optimal LX-mixer};
    \node [block] (connect) at ($(cost)!0.5!(out)$) {
        \begin{minipage}{\linewidth}
        \textbf{(4)} Find connected graph
        \phantom{\textbf{(4)}}
        composed by $G_{\lX_i,j}$ \\
        \phantom{\textbf{(4)}}
        with minimal cost
        \end{minipage}
        };
    
    \draw[->] (init) to (graphs);
    \draw[->] (graphs) to (subgraphs);
    \draw[->] ($(subgraphs.east) + (0,-0cm)$) to (cost);
    \draw[->] (cost) to (connect);
    \draw[->] (connect) to (out);
\end{tikzpicture}

%% file: pictures/statistics_cnots_n3_4.tex
\begin{tikzpicture}

\definecolor{darkgray176}{RGB}{176,176,176}
\definecolor{goldenrod1911910}{RGB}{191,191,0}
\definecolor{green}{RGB}{0,128,0}
\definecolor{green01270}{RGB}{0,127,0}
\definecolor{lightgray204}{RGB}{204,204,204}
\definecolor{yellow}{RGB}{255,255,0}

\begin{axis}[
legend cell align={left},
legend style={
  fill opacity=0.8,
  draw opacity=1,
  text opacity=1,
  at={(1.185,.98)},
  anchor=north west,
  draw=lightgray204
},
tick align=outside,
tick pos=left,
x grid style={darkgray176},
xlabel={\(\displaystyle |B|\)},
xmin=2, xmax=8,
xtick style={color=black},
y grid style={darkgray176},
ylabel={\#CNOTS},
ymin=., ymax=82,
ytick style={color=black},
title={$n=3$}
]

\path [draw=red, fill=red, opacity=0.35]
(axis cs:2,13.628261106855)
--(axis cs:2,8.05173889314505)
--(axis cs:3,17.6076277056389)
--(axis cs:4,28.1245895715788)
--(axis cs:5,37.7857384988589)
--(axis cs:6,49.0410435966406)
--(axis cs:7,60.2881155249713)
--(axis cs:8,72.080800080016)
--(axis cs:8,80.079199919984)
--(axis cs:8,80.079199919984)
--(axis cs:7,69.1518844750287)
--(axis cs:6,58.2389564033594)
--(axis cs:5,46.2142615011411)
--(axis cs:4,36.8354104284212)
--(axis cs:3,24.3123722943611)
--(axis cs:2,13.628261106855)
--cycle;

\path [draw=yellow, fill=yellow, opacity=0.35]
(axis cs:2,2.81413055342748)
--(axis cs:2,0.025869446572524)
--(axis cs:3,4.84646958760843)
--(axis cs:4,9.85171273060906)
--(axis cs:5,22.6319340664317)
--(axis cs:6,41.263461268284)
--(axis cs:7,60.2881155249713)
--(axis cs:8,72.080800080016)
--(axis cs:8,80.079199919984)
--(axis cs:8,80.079199919984)
--(axis cs:7,69.1518844750287)
--(axis cs:6,53.856538731716)
--(axis cs:5,31.3680659335683)
--(axis cs:4,17.1482872693909)
--(axis cs:3,10.0335304123916)
--(axis cs:2,2.81413055342748)
--cycle;

\path [draw=blue, fill=blue, opacity=0.35]
(axis cs:2,13.628261106855)
--(axis cs:2,8.05173889314505)
--(axis cs:3,16.0436133778645)
--(axis cs:4,15.02269381711)
--(axis cs:5,11.9488792270697)
--(axis cs:6,5.71376013954648)
--(axis cs:7,6)
--(axis cs:8,0)
--(axis cs:8,0)
--(axis cs:8,0)
--(axis cs:7,6)
--(axis cs:6,9.56623986045352)
--(axis cs:5,17.5711207729303)
--(axis cs:4,30.01730618289)
--(axis cs:3,21.7963866221355)
--(axis cs:2,13.628261106855)
--cycle;

\path [draw=green, fill=green, opacity=0.35]
(axis cs:2,2.81413055342748)
--(axis cs:2,0.025869446572524)
--(axis cs:3,4.29969699288215)
--(axis cs:4,4.66744873268133)
--(axis cs:5,11.9488792270697)
--(axis cs:6,4.00638419528068)
--(axis cs:7,6)
--(axis cs:8,0)
--(axis cs:8,0)
--(axis cs:8,0)
--(axis cs:7,6)
--(axis cs:6,9.51361580471931)
--(axis cs:5,17.5711207729303)
--(axis cs:4,11.3725512673187)
--(axis cs:3,8.26030300711785)
--(axis cs:2,2.81413055342748)
--cycle;

\addplot [semithick, red, mark=+, mark size=3, mark options={solid}]
table {%
2 10.84
3 20.96
4 32.48
5 42
6 53.64
7 64.72
8 76.08
};
\addlegendentry{$\mu\pm\sigma$, chain}

\addplot [semithick, goldenrod1911910, mark=star, mark size=3, mark options={solid}]
table {%
2 1.42
3 7.44
4 13.5
5 27
6 47.56
7 64.72
8 76.08
};
\addlegendentry{$\mu\pm\sigma$, chain restr.}

\addplot [semithick, blue, mark=x, mark size=3, mark options={solid}]
table {%
2 10.84
3 18.92
4 22.52
5 14.76
6 7.64
7 6
8 0
};
\addlegendentry{$\mu\pm\sigma$, optimal}

\addplot [semithick, green01270, mark=*, mark size=3, mark options={solid}]
table {%
2 1.42 
3 6.28 
4 8.02 
5 14.76 
6 6.76 
7 6
8 0
};
\addlegendentry{$\mu\pm\sigma$, optimal restr.}

\addplot [
    black, 
    fill=black, 
    opacity=0.15,
    mark=*, mark size=3, 
    mark options={solid}
]
coordinates {
    (2,4)
    (2,0)
    (3,4)
    (4,0)
    (5,12)
    (6,4)
    (7,6)
    (8,0)
    (8,0)
    (8,0)
    (7,6)
    (6,12)
    (5,20)
    (4,12)
    (3,8)
    (2,4)
};

\addplot [
    black, 
    fill=black, 
    dotted, 
    opacity=0.15,
    mark=x, mark size=3, 
    mark options={solid}
]
coordinates {
    (2,16)
    (2,8)
    (3,16)
    (4,4)
    (5,12)
    (6,6)
    (7,6)
    (8,0)
    (8,0)
    (8,0)
    (7,6)
    (6,12)
    (5,20)
    (4,28)
    (3,24)
    (2,16)
};

\addplot [
    black, 
    fill=black, 
    dotted, 
    opacity=0.15,
    mark=+, mark size=3, 
    mark options={solid}
]
coordinates {
    (2,16)
    (2,8)
    (3,16)
    (4,24)
    (5,36)
    (6,48)
    (7,60)
    (8,72)
    (8,80)
    (8,80)
    (7,76)
    (6,64)
    (5,52)
    (4,44)
    (3,28)
    (2,16)
};

\addplot [
    black, 
    fill=black, 
    dotted, 
    opacity=0.15,
    mark=star, mark size=3, 
    mark options={solid}
]
coordinates {
    (2,4)
    (2,0)
    (3,4)
    (4,6)
    (5,18)
    (6,32)
    (7,60)
    (8,72)
    (8,80)
    (8,80)
    (7,76)
    (6,64)
    (5,36)
    (4,22)
    (3,14)
    (2,4)
};

\end{axis}

\end{tikzpicture}

%% file: pictures/statistics_cnots_n4_4.tex
\begin{tikzpicture}

\definecolor{darkgray176}{RGB}{176,176,176}
\definecolor{goldenrod1911910}{RGB}{191,191,0}
\definecolor{green}{RGB}{0,128,0}
\definecolor{green01270}{RGB}{0,127,0}
\definecolor{lightgray204}{RGB}{204,204,204}
\definecolor{yellow}{RGB}{255,255,0}

\begin{axis}[
tick align=outside,
tick pos=left,
x grid style={darkgray176},
xlabel={\(\displaystyle |B|\)},
xmin=2, xmax=16,
xtick style={color=black},
y grid style={darkgray176},
ymin=-0, ymax=520,
ytick style={color=black},
title={$n=4$}
]

\path [draw=red, fill=red, opacity=0.35]
(axis cs:2,41.4025693713918)
--(axis cs:2,27.0774306286082)
--(axis cs:3,55.5750324675872)
--(axis cs:4,83.6143329997839)
--(axis cs:5,121.660455382849)
--(axis cs:6,153.184622963246)
--(axis cs:7,185.194225713315)
--(axis cs:8,218.578919259613)
--(axis cs:9,253.586165384931)
--(axis cs:10,290.402474424058)
--(axis cs:11,321.50239480306)
--(axis cs:12,356.813169130432)
--(axis cs:13,390.480781264902)
--(axis cs:14,429.026621759589)
--(axis cs:15,463.945014813942)
--(axis cs:16,492.362066380846)
--(axis cs:16,533.397933619154)
--(axis cs:16,533.397933619154)
--(axis cs:15,506.134985186058)
--(axis cs:14,472.573378240411)
--(axis cs:13,431.439218735098)
--(axis cs:12,394.706830869568)
--(axis cs:11,360.41760519694)
--(axis cs:10,321.437525575942)
--(axis cs:9,291.533834615069)
--(axis cs:8,253.741080740387)
--(axis cs:7,215.125774286685)
--(axis cs:6,184.415377036755)
--(axis cs:5,143.619544617151)
--(axis cs:4,107.265667000216)
--(axis cs:3,74.8249675324128)
--(axis cs:2,41.4025693713918)
--cycle;

\path [draw=yellow, fill=yellow, opacity=0.35]
(axis cs:2,4.35064234284795)
--(axis cs:2,0.769357657152049)
--(axis cs:3,6.35447981814378)
--(axis cs:4,11.3215002344483)
--(axis cs:5,22.4656441684057)
--(axis cs:6,41.989070910251)
--(axis cs:7,67.8262404585589)
--(axis cs:8,88.3810167097374)
--(axis cs:9,114.422892944016)
--(axis cs:10,145.664102728997)
--(axis cs:11,186.598903080409)
--(axis cs:12,258.803176481999)
--(axis cs:13,331.877500107305)
--(axis cs:14,406.883658791668)
--(axis cs:15,463.945014813942)
--(axis cs:16,492.362066380846)
--(axis cs:16,533.397933619154)
--(axis cs:16,533.397933619154)
--(axis cs:15,506.134985186058)
--(axis cs:14,459.356341208332)
--(axis cs:13,390.122499892695)
--(axis cs:12,318.236823518001)
--(axis cs:11,249.321096919591)
--(axis cs:10,185.535897271003)
--(axis cs:9,144.977107055984)
--(axis cs:8,112.378983290263)
--(axis cs:7,89.9337595414411)
--(axis cs:6,65.130929089749)
--(axis cs:5,35.8543558315943)
--(axis cs:4,20.9184997655517)
--(axis cs:3,14.7655201818562)
--(axis cs:2,4.35064234284795)
--cycle;

\path [draw=blue, fill=blue, opacity=0.35]
(axis cs:2,41.4025693713918)
--(axis cs:2,27.0774306286082)
--(axis cs:3,50.8142429604521)
--(axis cs:4,60.4067592947507)
--(axis cs:5,62.4430377674442)
--(axis cs:6,61.0885294816556)
--(axis cs:7,47.0842618567077)
--(axis cs:8,36.4032764404667)
--(axis cs:9,36.2280422434254)
--(axis cs:10,38.6050554810409)
--(axis cs:11,28.5124196594564)
--(axis cs:12,20.9948651889539)
--(axis cs:13,14.8480409959152)
--(axis cs:14,8.07222642852984)
--(axis cs:15,8)
--(axis cs:16,0)
--(axis cs:16,0)
--(axis cs:16,0)
--(axis cs:15,8)
--(axis cs:14,11.8877735714702)
--(axis cs:13,20.3119590040848)
--(axis cs:12,35.6451348110461)
--(axis cs:11,47.7275803405436)
--(axis cs:10,53.4749445189591)
--(axis cs:9,47.2119577565746)
--(axis cs:8,54.9567235595333)
--(axis cs:7,65.4757381432922)
--(axis cs:6,92.0314705183444)
--(axis cs:5,115.636962232556)
--(axis cs:4,95.4332407052493)
--(axis cs:3,67.4257570395479)
--(axis cs:2,41.4025693713918)
--cycle;

\path [draw=green, fill=green, opacity=0.35]
(axis cs:2,4.35064234284795)
--(axis cs:2,0.769357657152049)
--(axis cs:3,5.43948056328158)
--(axis cs:4,7.80841283665986)
--(axis cs:5,14.906796751063)
--(axis cs:6,18.8321918429495)
--(axis cs:7,17.6741324499133)
--(axis cs:8,20.2303559762716)
--(axis cs:9,26.7980072140203)
--(axis cs:10,25.6835549032153)
--(axis cs:11,20.2132215065865)
--(axis cs:12,11.4643229102421)
--(axis cs:13,14.6589132574292)
--(axis cs:14,6.72363887147455)
--(axis cs:15,8)
--(axis cs:16,0)
--(axis cs:16,0)
--(axis cs:16,0)
--(axis cs:15,8)
--(axis cs:14,11.9963611285255)
--(axis cs:13,20.1810867425708)
--(axis cs:12,21.1756770897579)
--(axis cs:11,30.5467784934135)
--(axis cs:10,33.5564450967847)
--(axis cs:9,38.4419927859797)
--(axis cs:8,38.4896440237284)
--(axis cs:7,39.0858675500867)
--(axis cs:6,30.0078081570505)
--(axis cs:5,20.373203248937)
--(axis cs:4,14.5515871633401)
--(axis cs:3,11.6005194367184)
--(axis cs:2,4.35064234284795)
--cycle;

\addplot [semithick, red, mark=+, mark size=3, mark options={solid}]
table {%
2 34.24
3 65.2
4 95.44
5 132.64
6 168.8
7 200.16
8 236.16
9 272.56
10 305.92
11 340.96
12 375.76
13 410.96
14 450.8
15 485.04
16 512.88
};

\addplot [semithick, goldenrod1911910, mark=star, mark size=3, mark options={solid}]
table {%
2 2.56
3 10.56
4 16.12
5 29.16
6 53.56
7 78.88
8 100.38
9 129.7
10 165.6
11 217.96
12 288.52
13 361
14 433.12
15 485.04
16 512.88
};

\addplot [semithick, blue, mark=x, mark size=3, mark options={solid}]
table {%
2 34.24
3 59.12
4 77.92
5 89.04
6 76.56
7 56.28
8 45.68
9 41.72
10 46.04
11 38.12
12 28.32
13 17.58
14 9.98
15 8
16 0
};

\addplot [semithick, green01270, mark=*, mark size=3, mark options={solid}]
table {%
2 2.56
3 8.52
4 11.18
5 17.64
6 24.42
7 28.38
8 29.36
9 32.62
10 29.62
11 25.38
12 16.32
13 17.42
14 9.36
15 8
16 0
};

\addplot [
    black, 
    fill=black, 
    opacity=0.15,
    mark=*, mark size=3, 
    mark options={solid}
]
coordinates {
    (2,6)
    (2,0)
    (3,4)
    (4,0)
    (5,12)
    (6,6)
    (7,6)
    (8,4)
    (9,22)
    (10,22)
    (11,16)
    (12,4)
    (13,14)
    (14,6)
    (15,8)
    (16,0)
    (16,0)
    (16,0)
    (15,8)
    (14,16)
    (13,22)
    (12,24)
    (11,34)
    (10,40)
    (9,48)
    (8,40)
    (7,48)
    (6,32)
    (5,24)
    (4,16)
    (3,12)
    (2,6)
};

\addplot [
    black, 
    fill=black, 
    opacity=0.15,
    mark=x, mark size=3, 
    mark options={solid}
]
coordinates {
    (2,48)
    (2,24)
    (3,48)
    (4,16)
    (5,40)
    (6,36)
    (7,24)
    (8,18)
    (9,32)
    (10,24)
    (11,20)
    (12,8)
    (13,14)
    (14,8)
    (15,8)
    (16,0)
    (16,0)
    (16,0)
    (15,8)
    (14,16)
    (13,22)
    (12,40)
    (11,52)
    (10,64)
    (9,56)
    (8,72)
    (7,72)
    (6,128)
    (5,136)
    (4,104)
    (3,72)
    (2,48)
};

\addplot [
    black, 
    fill=black, 
    opacity=0.15,
    mark=+, mark size=3, 
    mark options={solid}
]
coordinates {
    (2,48)
    (2,24)
    (3,48)
    (4,72)
    (5,104)
    (6,136)
    (7,168)
    (8,200)
    (9,232)
    (10,264)
    (11,296)
    (12,336)
    (13,360)
    (14,384)
    (15,440)
    (16,472)
    (16,576)
    (16,576)
    (15,528)
    (14,512)
    (13,464)
    (12,416)
    (11,384)
    (10,336)
    (9,328)
    (8,272)
    (7,232)
    (6,208)
    (5,160)
    (4,120)
    (3,88)
    (2,48)
};

\addplot [
    black, 
    fill=black, 
    opacity=0.15,
    mark=star, mark size=3, 
    mark options={solid}
]
coordinates {
    (2,6)
    (2,0)
    (3,4)
    (4,8)
    (5,14)
    (6,30)
    (7,54)
    (8,70)
    (9,98)
    (10,128)
    (11,156)
    (12,228)
    (13,288)
    (14,340)
    (15,440)
    (16,472)
    (16,576)
    (16,576)
    (15,528)
    (14,512)
    (13,448)
    (12,368)
    (11,300)
    (10,240)
    (9,172)
    (8,128)
    (7,108)
    (6,82)
    (5,54)
    (4,30)
    (3,22)
    (2,6)
};

\end{axis}

\end{tikzpicture}

%% file: pictures/X4.tex
\begin{tikzpicture}[scale=4]
\node[color=black,opacity=1] (1100) at (4.695/1.5,0) {$\scriptstyle\ket{1100}$};
\node[color=black,opacity=1] (0111) at (2.812/1.5,0) {$\scriptstyle\ket{0111}$};
\node[color=black,opacity=1] (1001) at (3.565/1.5,0) {$\scriptstyle\ket{1001}$};
\node[color=black,opacity=1] (0100) at (1.682/1.5,0) {$\scriptstyle\ket{0100}$};
\node[color=black,opacity=1] (1110) at (5.448/1.5,0) {$\scriptstyle\ket{1110}$};
\node[color=black,opacity=1] (1010) at (3.942/1.5,0) {$\scriptstyle\ket{1010}$};
\node[color=black,opacity=1] (1101) at (5.072/1.5,0) {$\scriptstyle\ket{1101}$};
\node[color=black,opacity=1] (1011) at (4.318/1.5,0) {$\scriptstyle\ket{1011}$};
\node[color=black,opacity=1] (0101) at (2.058/1.5,0) {$\scriptstyle\ket{0101}$};
\node[color=black,opacity=1] (0000) at (0.175/1.5,0) {$\scriptstyle\ket{0000}$};
\node[color=black,opacity=1] (1000) at (3.188/1.5,0) {$\scriptstyle\ket{1000}$};
\node[color=black,opacity=1] (0010) at (0.928/1.5,0) {$\scriptstyle\ket{0010}$};
\node[color=black,opacity=1] (0011) at (1.305/1.5,0) {$\scriptstyle\ket{0011}$};
\node[color=black,opacity=1] (0001) at (0.552/1.5,0) {$\scriptstyle\ket{0001}$};
\node[color=black,opacity=1] (1111) at (5.825/1.5,0) {$\scriptstyle\ket{1111}$};
\node[color=black,opacity=1] (0110) at (2.435/1.5,0) {$\scriptstyle\ket{0110}$};
\Edge[,color={58.6004535,76.17308133,192.189204015},bend=-56.31,Direct,RGB](1100)(0100)
\Edge[,color={58.6004535,76.17308133,192.189204015},bend=-56.31,Direct,RGB](0111)(1111)
\Edge[,color={58.6004535,76.17308133,192.189204015},bend=-56.31,Direct,RGB](1001)(0001)
\Edge[,color={58.6004535,76.17308133,192.189204015},bend=-56.31,Direct,RGB](0100)(1100)
\Edge[,color={58.6004535,76.17308133,192.189204015},bend=-56.31,Direct,RGB](1110)(0110)
\Edge[,color={58.6004535,76.17308133,192.189204015},bend=-56.31,Direct,RGB](1010)(0010)
\Edge[,color={58.6004535,76.17308133,192.189204015},bend=-56.31,Direct,RGB](1101)(0101)
\Edge[,color={58.6004535,76.17308133,192.189204015},bend=-56.31,Direct,RGB](1011)(0011)
\Edge[,color={58.6004535,76.17308133,192.189204015},bend=-56.31,Direct,RGB](0101)(1101)
\Edge[,color={58.6004535,76.17308133,192.189204015},bend=-56.31,Direct,RGB](0000)(1000)
\Edge[,color={58.6004535,76.17308133,192.189204015},bend=-56.31,Direct,RGB](1000)(0000)
\Edge[,color={58.6004535,76.17308133,192.189204015},bend=-56.31,Direct,RGB](0010)(1010)
\Edge[,color={58.6004535,76.17308133,192.189204015},bend=-56.31,Direct,RGB](0011)(1011)
\Edge[,color={58.6004535,76.17308133,192.189204015},bend=-56.31,Direct,RGB](0001)(1001)
\Edge[,color={58.6004535,76.17308133,192.189204015},bend=-56.31,Direct,RGB](1111)(0111)
\Edge[,color={58.6004535,76.17308133,192.189204015},bend=-56.31,Direct,RGB](0110)(1110)
\Edge[,color={170.14949570500002,198.68999653499998,253.204599315},bend=-56.31,Direct,RGB](1100)(1110)
\Edge[,color={170.14949570500002,198.68999653499998,253.204599315},bend=-56.31,Direct,RGB](0111)(0101)
\Edge[,color={170.14949570500002,198.68999653499998,253.204599315},bend=-56.31,Direct,RGB](1001)(1011)
\Edge[,color={170.14949570500002,198.68999653499998,253.204599315},bend=-56.31,Direct,RGB](0100)(0110)
\Edge[,color={170.14949570500002,198.68999653499998,253.204599315},bend=-56.31,Direct,RGB](1110)(1100)
\Edge[,color={170.14949570500002,198.68999653499998,253.204599315},bend=-56.31,Direct,RGB](1010)(1000)
\Edge[,color={170.14949570500002,198.68999653499998,253.204599315},bend=-56.31,Direct,RGB](1101)(1111)
\Edge[,color={170.14949570500002,198.68999653499998,253.204599315},bend=-56.31,Direct,RGB](1011)(1001)
\Edge[,color={170.14949570500002,198.68999653499998,253.204599315},bend=-56.31,Direct,RGB](0101)(0111)
\Edge[,color={170.14949570500002,198.68999653499998,253.204599315},bend=-56.31,Direct,RGB](0000)(0010)
\Edge[,color={170.14949570500002,198.68999653499998,253.204599315},bend=-56.31,Direct,RGB](1000)(1010)
\Edge[,color={170.14949570500002,198.68999653499998,253.204599315},bend=-56.31,Direct,RGB](0010)(0000)
\Edge[,color={170.14949570500002,198.68999653499998,253.204599315},bend=-56.31,Direct,RGB](0011)(0001)
\Edge[,color={170.14949570500002,198.68999653499998,253.204599315},bend=-56.31,Direct,RGB](0001)(0011)
\Edge[,color={170.14949570500002,198.68999653499998,253.204599315},bend=-56.31,Direct,RGB](1111)(1101)
\Edge[,color={170.14949570500002,198.68999653499998,253.204599315},bend=-56.31,Direct,RGB](0110)(0100)
\Edge[,color={246.891866745,183.8152455,156.13471279},bend=-56.31,Direct,RGB](1100)(1101)
\Edge[,color={246.891866745,183.8152455,156.13471279},bend=-56.31,Direct,RGB](0111)(0110)
\Edge[,color={246.891866745,183.8152455,156.13471279},bend=-56.31,Direct,RGB](1001)(1000)
\Edge[,color={246.891866745,183.8152455,156.13471279},bend=-56.31,Direct,RGB](0100)(0101)
\Edge[,color={246.891866745,183.8152455,156.13471279},bend=-56.31,Direct,RGB](1110)(1111)
\Edge[,color={246.891866745,183.8152455,156.13471279},bend=-56.31,Direct,RGB](1010)(1011)
\Edge[,color={246.891866745,183.8152455,156.13471279},bend=-56.31,Direct,RGB](1101)(1100)
\Edge[,color={246.891866745,183.8152455,156.13471279},bend=-56.31,Direct,RGB](1011)(1010)
\Edge[,color={246.891866745,183.8152455,156.13471279},bend=-56.31,Direct,RGB](0101)(0100)
\Edge[,color={246.891866745,183.8152455,156.13471279},bend=-56.31,Direct,RGB](0000)(0001)
\Edge[,color={246.891866745,183.8152455,156.13471279},bend=-56.31,Direct,RGB](1000)(1001)
\Edge[,color={246.891866745,183.8152455,156.13471279},bend=-56.31,Direct,RGB](0010)(0011)
\Edge[,color={246.891866745,183.8152455,156.13471279},bend=-56.31,Direct,RGB](0011)(0010)
\Edge[,color={246.891866745,183.8152455,156.13471279},bend=-56.31,Direct,RGB](0001)(0000)
\Edge[,color={246.891866745,183.8152455,156.13471279},bend=-56.31,Direct,RGB](1111)(1110)
\Edge[,color={246.891866745,183.8152455,156.13471279},bend=-56.31,Direct,RGB](0110)(0111)
\Edge[,color={179.94665529,3.9668208,38.30936706},bend=-56.31,Direct,RGB](1100)(1000)
\Edge[,color={179.94665529,3.9668208,38.30936706},bend=-56.31,Direct,RGB](0111)(0011)
\Edge[,color={179.94665529,3.9668208,38.30936706},bend=-56.31,Direct,RGB](1001)(1101)
\Edge[,color={179.94665529,3.9668208,38.30936706},bend=-56.31,Direct,RGB](0100)(0000)
\Edge[,color={179.94665529,3.9668208,38.30936706},bend=-56.31,Direct,RGB](1110)(1010)
\Edge[,color={179.94665529,3.9668208,38.30936706},bend=-56.31,Direct,RGB](1010)(1110)
\Edge[,color={179.94665529,3.9668208,38.30936706},bend=-56.31,Direct,RGB](1101)(1001)
\Edge[,color={179.94665529,3.9668208,38.30936706},bend=-56.31,Direct,RGB](1011)(1111)
\Edge[,color={179.94665529,3.9668208,38.30936706},bend=-56.31,Direct,RGB](0101)(0001)
\Edge[,color={179.94665529,3.9668208,38.30936706},bend=-56.31,Direct,RGB](0000)(0100)
\Edge[,color={179.94665529,3.9668208,38.30936706},bend=-56.31,Direct,RGB](1000)(1100)
\Edge[,color={179.94665529,3.9668208,38.30936706},bend=-56.31,Direct,RGB](0010)(0110)
\Edge[,color={179.94665529,3.9668208,38.30936706},bend=-56.31,Direct,RGB](0011)(0111)
\Edge[,color={179.94665529,3.9668208,38.30936706},bend=-56.31,Direct,RGB](0001)(0101)
\Edge[,color={179.94665529,3.9668208,38.30936706},bend=-56.31,Direct,RGB](1111)(1011)
\Edge[,color={179.94665529,3.9668208,38.30936706},bend=-56.31,Direct,RGB](0110)(0010)

\definecolor{XIII}{RGB}{58.6004535,76.17308133,192.189204015};
\draw[XIII,very thick] (0.5,-.7) node[right,xshift=.9cm] {$XIII\langle IIII\rangle$} -- (0.7,-.7);

\definecolor{IIXI}{RGB}{170.14949570500002,198.68999653499998,253.204599315};
\draw[IIXI,very thick] (1.3,-.7) node[right,xshift=.9cm] {$IIXI\langle IIII\rangle$} -- (1.5,-.7);

\definecolor{IIIX}{RGB}{246.891866745,183.8152455,156.13471279};
\draw[IIIX,very thick] (2.1,-.7) node[right,xshift=.9cm] {$IIIX\langle IIII\rangle$} -- (2.3,-.7);

\definecolor{IXII}{RGB}{179.94665529,3.9668208,38.30936706};
\draw[IXII,very thick] (2.9,-.7) node[right,xshift=.9cm] {$IXII\langle IIII\rangle$} -- (3.1,-.7);

\end{tikzpicture}

%% file: pictures/XY_64.tex
\begin{tikzpicture}[scale=4]
\node[color=black,opacity=1] (001111) at (0.175/1.5,0) {$\scriptscriptstyle\ket{001111}$};
\node[color=black,opacity=1] (010111) at (0.579/1.5,0) {$\scriptscriptstyle\ket{010111}$};
\node[color=black,opacity=1] (011011) at (0.982/1.5,0) {$\scriptscriptstyle\ket{011011}$};
\node[color=black,opacity=1] (011101) at (1.386/1.5,0) {$\scriptscriptstyle\ket{011101}$};
\node[color=black,opacity=1] (011110) at (1.789/1.5,0) {$\scriptscriptstyle\ket{011110}$};
\node[color=black,opacity=1] (100111) at (2.193/1.5,0) {$\scriptscriptstyle\ket{100111}$};
\node[color=black,opacity=1] (101011) at (2.596/1.5,0) {$\scriptscriptstyle\ket{101011}$};
\node[color=black,opacity=1] (101101) at (3.000/1.5,0) {$\scriptscriptstyle\ket{101101}$};
\node[color=black,opacity=1] (101110) at (3.404/1.5,0) {$\scriptscriptstyle\ket{101110}$};
\node[color=black,opacity=1] (110011) at (3.807/1.5,0) {$\scriptscriptstyle\ket{110011}$};
\node[color=black,opacity=1] (110101) at (4.211/1.5,0) {$\scriptscriptstyle\ket{110101}$};
\node[color=black,opacity=1] (110110) at (4.614/1.5,0) {$\scriptscriptstyle\ket{110110}$};
\node[color=black,opacity=1] (111001) at (5.018/1.5,0) {$\scriptscriptstyle\ket{111001}$};
\node[color=black,opacity=1] (111010) at (5.421/1.5,0) {$\scriptscriptstyle\ket{111010}$};
\node[color=black,opacity=1] (111100) at (5.825/1.5,0) {$\scriptscriptstyle\ket{111100}$};

\Edge[,color={58.6004535,76.17308133,192.189204015},bend=-56.31,Direct,RGB](001111)(010111)
\Edge[,color={58.6004535,76.17308133,192.189204015},bend=-56.31,Direct,RGB](010111)(001111)
\Edge[,color={58.6004535,76.17308133,192.189204015},bend=-56.31,Direct,RGB](101011)(110011)
\Edge[,color={58.6004535,76.17308133,192.189204015},bend=-56.31,Direct,RGB](101101)(110101)
\Edge[,color={58.6004535,76.17308133,192.189204015},bend=-56.31,Direct,RGB](101110)(110110)
\Edge[,color={58.6004535,76.17308133,192.189204015},bend=-56.31,Direct,RGB](110011)(101011)https://www.overleaf.com/project/62e25ce55dd7026148648d5b
\Edge[,color={58.6004535,76.17308133,192.189204015},bend=-56.31,Direct,RGB](110101)(101101)
\Edge[,color={58.6004535,76.17308133,192.189204015},bend=-56.31,Direct,RGB](110110)(101110)

\definecolor{IXXIII}{RGB}{58.6004535,76.17308133,192.189204015};
\draw[IXXIII,very thick] (0.0,-.4) node[right,xshift=.7cm] {$X_2 X_3\langle -Z_2 Z_3\rangle$} -- (0.2,-.4);

\def\shift{.8}

\Edge[,color={141.34952682799997,175.97473785999998,253.8564648},bend=-56.31,Direct,RGB](010111)(011011)
\Edge[,color={141.34952682799997,175.97473785999998,253.8564648},bend=-56.31,Direct,RGB](011011)(010111)
\Edge[,color={141.34952682799997,175.97473785999998,253.8564648},bend=-56.31,Direct,RGB](100111)(101011)
\Edge[,color={141.34952682799997,175.97473785999998,253.8564648},bend=-56.31,Direct,RGB](101011)(100111)
\Edge[,color={141.34952682799997,175.97473785999998,253.8564648},bend=-56.31,Direct,RGB](110101)(111001)
\Edge[,color={141.34952682799997,175.97473785999998,253.8564648},bend=-56.31,Direct,RGB](110110)(111010)
\Edge[,color={141.34952682799997,175.97473785999998,253.8564648},bend=-56.31,Direct,RGB](111001)(110101)
\Edge[,color={141.34952682799997,175.97473785999998,253.8564648},bend=-56.31,Direct,RGB](111010)(110110)

\definecolor{IIXXII}{RGB}{141.34952682799997,175.97473785999998,253.8564648}
\draw[IIXXII,very thick] (\shift, -.4) node[right,xshift=.7cm] {$X_3 X_4\langle -Z_3 Z_4\rangle$} -- (\shift+.2,-.4);

\Edge[,color={21,21,21},bend=-56.31,Direct,RGB](010111)(100111)
\Edge[,color={21,21,21},bend=-56.31,Direct,RGB](011011)(101011)
\Edge[,color={21,21,21},bend=-56.31,Direct,RGB](011101)(101101)
\Edge[,color={21,21,21},bend=-56.31,Direct,RGB](011110)(101110)
\Edge[,color={21,21,21},bend=-56.31,Direct,RGB](100111)(010111)
\Edge[,color={21,21,21},bend=-56.31,Direct,RGB](101011)(011011)
\Edge[,color={21,21,21},bend=-56.31,Direct,RGB](101101)(011101)
\Edge[,color={21,21,21},bend=-56.31,Direct,RGB](101110)(011110)

\definecolor{XXIIII}{RGB}{21,21,21}
\draw[XXIIII,very thick] (2*\shift,-.4) node[right,xshift=.7cm] {$X_1 X_2\langle -Z_1 Z_2\rangle$} -- (2*\shift+.2,-.4);

\Edge[,color={243.946568799,152.498624793,121.71208455},bend=-56.31,Direct,RGB](011011)(011101)
\Edge[,color={243.946568799,152.498624793,121.71208455},bend=-56.31,Direct,RGB](011101)(011011)
\Edge[,color={243.946568799,152.498624793,121.71208455},bend=-56.31,Direct,RGB](101011)(101101)
\Edge[,color={243.946568799,152.498624793,121.71208455},bend=-56.31,Direct,RGB](101101)(101011)
\Edge[,color={243.946568799,152.498624793,121.71208455},bend=-56.31,Direct,RGB](110011)(110101)
\Edge[,color={243.946568799,152.498624793,121.71208455},bend=-56.31,Direct,RGB](110101)(110011)
\Edge[,color={243.946568799,152.498624793,121.71208455},bend=-56.31,Direct,RGB](111010)(111100)
\Edge[,color={243.946568799,152.498624793,121.71208455},bend=-56.31,Direct,RGB](111100)(111010)

\definecolor{IIIXXI}{RGB}{243.946568799,152.498624793,121.71208455}
\draw[IIIXXI,very thick] (3*\shift,-.4) node[right,xshift=.7cm] {$X_4 X_5\langle -Z_4 Z_5\rangle$} -- (3*\shift+.2,-.4);

\Edge[,color={179.94665529,3.9668208,38.30936706},bend=-56.31,Direct,RGB](011101)(011110)
\Edge[,color={179.94665529,3.9668208,38.30936706},bend=-56.31,Direct,RGB](011110)(011101)
\Edge[,color={179.94665529,3.9668208,38.30936706},bend=-56.31,Direct,RGB](101101)(101110)
\Edge[,color={179.94665529,3.9668208,38.30936706},bend=-56.31,Direct,RGB](101110)(101101)
\Edge[,color={179.94665529,3.9668208,38.30936706},bend=-56.31,Direct,RGB](110101)(110110)
\Edge[,color={179.94665529,3.9668208,38.30936706},bend=-56.31,Direct,RGB](110110)(110101)
\Edge[,color={179.94665529,3.9668208,38.30936706},bend=-56.31,Direct,RGB](111001)(111010)
\Edge[,color={179.94665529,3.9668208,38.30936706},bend=-56.31,Direct,RGB](111010)(111001)

\definecolor{IIIIXX}{RGB}{179.94665529,3.9668208,38.30936706}
\draw[IIIIXX,very thick] (4*\shift,-.4) node[right,xshift=.7cm] {$X_5 X_6\langle -Z_5 Z_6\rangle$} -- (4*\shift+.2,-.4);

\end{tikzpicture}

%% file: pictures/ex2_family.tex
\begin{tikzpicture}[scale=.60]%

\definecolor{IXII}{RGB}{58.6004535,76.17308133,192.189204015};
\definecolor{IIXX}{RGB}{170.14949570500002,198.68999653499998,253.204599315};
\definecolor{XIXI}{RGB}{246.891866745,183.8152455,156.13471279};
\definecolor{XIII}{RGB}{179.94665529,3.9668208,38.30936706};

\node[] at (0,10)  {$G_{IIXI}$};
\node[] at (2,10)  {$G_{IIXX}$};
\node[] at (4,10)  {$G_{IXII}$};
\node[] at (6,10)  {$G_{IXIX}$};
\node[] at (8,10)  {$G_{IXXX}$};
\node[] at (10,10) {$G_{XIII}$};
\node[] at (12,10) {$G_{XIIX}$};
\node[] at (14,10) {$G_{XIXI}$};
\node[] at (16,10) {$G_{XIXX}$};
\node[] at (18,10) {$G_{XXII}$};
\node[] at (20,10) {$G_{XXIX}$};
\node[] at (22,10) {$G_{XXXI}$};
\node[] at (24,10) {$G_{XXXX}$};

\def\x{0}
\Vertex[style={color=gray},x=\x,y=3  ,label=1010]{1010}
\Vertex[style={color=gray},x=\x,y=6  ,label=0111]{0111}
\Vertex[style={color=gray},x=\x,y=0  ,label=1110]{1110}
\Vertex[style={color=gray},x=\x,y=4.5,label=1001]{1001}
\Vertex[style={color=gray},x=\x,y=7.5,label=0010]{0010}
\Vertex[style={color=gray},x=\x,y=9  ,label=0000]{0000}
\Vertex[style={color=gray},x=\x,y=1.5,label=1101]{1101}
\Edge[bend=+45](0010)(0000)

\def\x{2}
\Vertex[style={color=gray,opacity=.5},x=\x,y=3  ,label=1010]{1010}
\Vertex[style={color=gray,opacity=.5},x=\x,y=6  ,label=0111]{0111}
\Vertex[style={color=gray,opacity=.5},x=\x,y=0  ,label=1110]{1110}
\Vertex[style={color=gray,opacity=.5},x=\x,y=4.5,label=1001]{1001}
\Vertex[style={color=gray,opacity=.5},x=\x,y=7.5,label=0010]{0010}
\Vertex[style={color=gray,opacity=.5},x=\x,y=9  ,label=0000]{0000}
\Vertex[style={color=gray,opacity=.5},x=\x,y=1.5,label=1101]{1101}
\Edge[color=IIXX,bend=+45](1010)(1001)
\Edge[color=IIXX,bend=+45](1110)(1101)

\def\x{4}
\Vertex[style={color=gray},x=\x,y=3  ,label=1010]{1010}
\Vertex[style={color=gray},x=\x,y=6  ,label=0111]{0111}
\Vertex[style={color=gray},x=\x,y=0  ,label=1110]{1110}
\Vertex[style={color=gray},x=\x,y=4.5,label=1001]{1001}
\Vertex[style={color=gray},x=\x,y=7.5,label=0010]{0010}
\Vertex[style={color=gray},x=\x,y=9  ,label=0000]{0000}
\Vertex[style={color=gray},x=\x,y=1.5,label=1101]{1101}
\Edge[color=IXII,bend=-45](1010)(1110)
\Edge[color=IXII,bend=+45](1001)(1101)

\def\x{6}
\Vertex[style={color=gray,opacity=.5},x=\x,y=3  ,label=1010]{1010}
\Vertex[style={color=gray,opacity=.5},x=\x,y=6  ,label=0111]{0111}
\Vertex[style={color=gray,opacity=.5},x=\x,y=0  ,label=1110]{1110}
\Vertex[style={color=gray,opacity=.5},x=\x,y=4.5,label=1001]{1001}
\Vertex[style={color=gray,opacity=.5},x=\x,y=7.5,label=0010]{0010}
\Vertex[style={color=gray,opacity=.5},x=\x,y=9  ,label=0000]{0000}
\Vertex[style={color=gray,opacity=.5},x=\x,y=1.5,label=1101]{1101}
\Edge[bend=+45](0111)(0010)

\def\x{8}
\Vertex[style={color=gray},x=\x,y=3  ,label=1010]{1010}
\Vertex[style={color=gray},x=\x,y=6  ,label=0111]{0111}
\Vertex[style={color=gray},x=\x,y=0  ,label=1110]{1110}
\Vertex[style={color=gray},x=\x,y=4.5,label=1001]{1001}
\Vertex[style={color=gray},x=\x,y=7.5,label=0010]{0010}
\Vertex[style={color=gray},x=\x,y=9  ,label=0000]{0000}
\Vertex[style={color=gray},x=\x,y=1.5,label=1101]{1101}
\Edge[bend=+45](1010)(1101)
\Edge[bend=+45](0111)(0000)
\Edge[bend=+45](1110)(1001)

\def\x{10}
\Vertex[style={color=gray,opacity=.5},x=\x,y=3  ,label=1010]{1010}
\Vertex[style={color=gray,opacity=.5},x=\x,y=6  ,label=0111]{0111}
\Vertex[style={color=gray,opacity=.5},x=\x,y=0  ,label=1110]{1110}
\Vertex[style={color=gray,opacity=.5},x=\x,y=4.5,label=1001]{1001}
\Vertex[style={color=gray,opacity=.5},x=\x,y=7.5,label=0010]{0010}
\Vertex[style={color=gray,opacity=.5},x=\x,y=9  ,label=0000]{0000}
\Vertex[style={color=gray,opacity=.5},x=\x,y=1.5,label=1101]{1101}
\Edge[color=XIII,bend=+35](1010)(0010)

\def\x{12}
\Vertex[style={color=gray},x=\x,y=3  ,label=1010]{1010}
\Vertex[style={color=gray},x=\x,y=6  ,label=0111]{0111}
\Vertex[style={color=gray},x=\x,y=0  ,label=1110]{1110}
\Vertex[style={color=gray},x=\x,y=4.5,label=1001]{1001}
\Vertex[style={color=gray},x=\x,y=7.5,label=0010]{0010}
\Vertex[style={color=gray},x=\x,y=9  ,label=0000]{0000}
\Vertex[style={color=gray},x=\x,y=1.5,label=1101]{1101}
\Edge[bend=+35](0111)(1110)
\Edge[bend=+35](1001)(0000)

\def\x{14}
\Vertex[style={color=gray,opacity=.5},x=\x,y=3  ,label=1010]{1010}
\Vertex[style={color=gray,opacity=.5},x=\x,y=6  ,label=0111]{0111}
\Vertex[style={color=gray,opacity=.5},x=\x,y=0  ,label=1110]{1110}
\Vertex[style={color=gray,opacity=.5},x=\x,y=4.5,label=1001]{1001}
\Vertex[style={color=gray,opacity=.5},x=\x,y=7.5,label=0010]{0010}
\Vertex[style={color=gray,opacity=.5},x=\x,y=9  ,label=0000]{0000}
\Vertex[style={color=gray,opacity=.5},x=\x,y=1.5,label=1101]{1101}
\Edge[color=XIXI,bend=+35](1010)(0000)
\Edge[color=XIXI,bend=+35](0111)(1101)

\def\x{16}
\Vertex[style={color=gray},x=\x,y=3  ,label=1010]{1010}
\Vertex[style={color=gray},x=\x,y=6  ,label=0111]{0111}
\Vertex[style={color=gray},x=\x,y=0  ,label=1110]{1110}
\Vertex[style={color=gray},x=\x,y=4.5,label=1001]{1001}
\Vertex[style={color=gray},x=\x,y=7.5,label=0010]{0010}
\Vertex[style={color=gray},x=\x,y=9  ,label=0000]{0000}
\Vertex[style={color=gray},x=\x,y=1.5,label=1101]{1101}
\Edge[bend=+45](1001)(0010)

\def\x{18}
\Vertex[style={color=gray,opacity=.5},x=\x,y=3  ,label=1010]{1010}
\Vertex[style={color=gray,opacity=.5},x=\x,y=6  ,label=0111]{0111}
\Vertex[style={color=gray,opacity=.5},x=\x,y=0  ,label=1110]{1110}
\Vertex[style={color=gray,opacity=.5},x=\x,y=4.5,label=1001]{1001}
\Vertex[style={color=gray,opacity=.5},x=\x,y=7.5,label=0010]{0010}
\Vertex[style={color=gray,opacity=.5},x=\x,y=9  ,label=0000]{0000}
\Vertex[style={color=gray,opacity=.5},x=\x,y=1.5,label=1101]{1101}
\Edge[bend=+30](1110)(0010)

\def\x{20}
\Vertex[style={color=gray},x=\x,y=3  ,label=1010]{1010}
\Vertex[style={color=gray},x=\x,y=6  ,label=0111]{0111}
\Vertex[style={color=gray},x=\x,y=0  ,label=1110]{1110}
\Vertex[style={color=gray},x=\x,y=4.5,label=1001]{1001}
\Vertex[style={color=gray},x=\x,y=7.5,label=0010]{0010}
\Vertex[style={color=gray},x=\x,y=9  ,label=0000]{0000}
\Vertex[style={color=gray},x=\x,y=1.5,label=1101]{1101}
\Edge[bend=+45](1010)(0111)
\Edge[bend=+25](0000)(1101)

\def\x{22}
\Vertex[style={color=gray,opacity=.5},x=\x,y=3  ,label=1010]{1010}
\Vertex[style={color=gray,opacity=.5},x=\x,y=6  ,label=0111]{0111}
\Vertex[style={color=gray,opacity=.5},x=\x,y=0  ,label=1110]{1110}
\Vertex[style={color=gray,opacity=.5},x=\x,y=4.5,label=1001]{1001}
\Vertex[style={color=gray,opacity=.5},x=\x,y=7.5,label=0010]{0010}
\Vertex[style={color=gray,opacity=.5},x=\x,y=9  ,label=0000]{0000}
\Vertex[style={color=gray,opacity=.5},x=\x,y=1.5,label=1101]{1101}
\Edge[bend=+45](0111)(1001)
\Edge[bend=-25](1110)(0000)

\def\x{24}
\Vertex[style={color=gray},x=\x,y=3  ,label=1010]{1010}
\Vertex[style={color=gray},x=\x,y=6  ,label=0111]{0111}
\Vertex[style={color=gray},x=\x,y=0  ,label=1110]{1110}
\Vertex[style={color=gray},x=\x,y=4.5,label=1001]{1001}
\Vertex[style={color=gray},x=\x,y=7.5,label=0010]{0010}
\Vertex[style={color=gray},x=\x,y=9  ,label=0000]{0000}
\Vertex[style={color=gray},x=\x,y=1.5,label=1101]{1101}
\Edge[bend=+35](0010)(1101)

\end{tikzpicture}

%% file: pictures/ex2.tex
\begin{tikzpicture}[scale=4]
\node[color=black,opacity=1] (0000) at (0.175/1.5,0) {$\ket{0000}$};
\node[color=black,opacity=1] (0010) at (1.117/1.5,0) {$\ket{0010}$};
\node[color=black,opacity=1] (0111) at (2.058/1.5,0) {$\ket{0111}$};
\node[color=black,opacity=1] (1001) at (3.000/1.5,0) {$\ket{1001}$};
\node[color=black,opacity=1] (1010) at (3.942/1.5,0) {$\ket{1010}$};
\node[color=black,opacity=1] (1101) at (4.883/1.5,0) {$\ket{1101}$};
\node[color=black,opacity=1] (1110) at (5.825/1.5,0) {$\ket{1110}$};

\Edge[,color={58.6004535,76.17308133,192.189204015},bend=-26.31,Direct,RGB](1101)(1001)
\Edge[,color={58.6004535,76.17308133,192.189204015},bend=-26.31,Direct,RGB](1001)(1101)
\Edge[,color={58.6004535,76.17308133,192.189204015},bend=-26.31,Direct,RGB](1110)(1010)
\Edge[,color={58.6004535,76.17308133,192.189204015},bend=-26.31,Direct,RGB](1010)(1110)
\Edge[,color={170.14949570500002,198.68999653499998,253.204599315},bend=-26.31,Direct,RGB](1101)(1110)
\Edge[,color={170.14949570500002,198.68999653499998,253.204599315},bend=-26.31,Direct,RGB](1001)(1010)
\Edge[,color={170.14949570500002,198.68999653499998,253.204599315},bend=-26.31,Direct,RGB](1110)(1101)
\Edge[,color={170.14949570500002,198.68999653499998,253.204599315},bend=-26.31,Direct,RGB](1010)(1001)
\Edge[,color={246.891866745,183.8152455,156.13471279},bend=-26.31,Direct,RGB](1101)(0111)
\Edge[,color={246.891866745,183.8152455,156.13471279},bend=-22.31,Direct,RGB](1010)(0000)
\Edge[,color={246.891866745,183.8152455,156.13471279},bend=-22.31,Direct,RGB](0000)(1010)
\Edge[,color={246.891866745,183.8152455,156.13471279},bend=-26.31,Direct,RGB](0111)(1101)
\Edge[,color={179.94665529,3.9668208,38.30936706},bend=-26.31,Direct,RGB](1010)(0010)
\Edge[,color={179.94665529,3.9668208,38.30936706},bend=-26.31,Direct,RGB](0010)(1010)

\def\dist{-.4}
\definecolor{XIXI}{RGB}{246.891866745,183.8152455,156.13471279};
\draw[XIXI,very thick] (0.1,\dist) node[right,xshift=.75cm] {$XIXI\langle +ZZZI \rangle$} -- (0.3,\dist);

\definecolor{XIII}{RGB}{179.94665529,3.9668208,38.30936706};
\draw[XIII,very thick] (1.05,\dist) node[right,xshift=.75cm] {$XIII\langle -IIZI, +IZII \rangle$} -- (1.25,\dist);

\definecolor{IXII}{RGB}{58.6004535,76.17308133,192.189204015};
\draw[IXII,very thick] (2.5-.25,\dist) node[right,xshift=.75cm] {$IXII\langle -ZIII \rangle$} -- (2.7-.25,\dist);

\definecolor{IIXX}{RGB}{170.14949570500002,198.68999653499998,253.204599315};
\draw[IIXX,very thick] (3.1,\dist) node[right,xshift=.75cm] {$IIXX\langle -ZIII \rangle$} -- (3.3,\dist);

\end{tikzpicture}

%% file: pictures/ex2_chain.tex
\begin{tikzpicture}[scale=4]
\node[opacity=0] (s) at (0.175/1.5,.5) {s};
\node[color=black,opacity=1] (0000) at (0.175/1.5,0) {$\ket{0000}$};
\node[color=black,opacity=1] (0010) at (1.117/1.5,0) {$\ket{0010}$};
\node[color=black,opacity=1] (0111) at (2.058/1.5,0) {$\ket{0111}$};
\node[color=black,opacity=1] (1001) at (3.000/1.5,0) {$\ket{1001}$};
\node[color=black,opacity=1] (1010) at (3.942/1.5,0) {$\ket{1010}$};
\node[color=black,opacity=1] (1101) at (4.883/1.5,0) {$\ket{1101}$};
\node[color=black,opacity=1] (1110) at (5.825/1.5,0) {$\ket{1110}$};

\draw[very thick, Latex-] (0010) to [bend right=-56.31] node[below] {$\scriptstyle IIXI\langle +ZIII, +IIIZ \rangle$} (0000);
\draw[very thick, Latex-] (0000) to [bend right=-56.31] (0010);

\draw[very thick, Latex-,gray] (0010) to [bend right=-56.31] node[above] {\color{gray} $\scriptstyle IXIX\langle +ZIII, -IIZI \rangle$} (0111);
\draw[very thick, Latex-,gray] (0111) to [bend right=-56.31] (0010);

\draw[very thick, Latex-] (0111) to [bend right=-56.31] (1001);
\draw[very thick, Latex-] (1001) to [bend right=-56.31] node[below] {$\scriptstyle XXXI\langle -ZIZI, +IZZI \rangle$} (0111);

\draw[very thick, Latex-,gray] (1001) to [bend right=-56.31] node[above] {\color{gray} $\scriptstyle IIXX\langle -ZZII, -IIZZ \rangle$} (1010);
\draw[very thick, Latex-,gray] (1010) to [bend right=-56.31] (1001);

\draw[very thick, Latex-] (1010) to [bend right=-56.31] (1101);
\draw[very thick, Latex-] (1101) to [bend right=-56.31] node[below] {$\scriptstyle IXXX\langle -ZZIZ \rangle$} (1010);

\draw[very thick, Latex-,gray] (1101) to [bend right=-56.31] node[above] {\color{gray} $\scriptstyle IIXX \langle -IZII, -IIZZ \rangle$} (1110);
\draw[very thick, Latex-,gray] (1110) to [bend right=-56.31] (1101);

\end{tikzpicture}

%% file: pictures/cartesian_composition.tex
\begin{tikzpicture}
            \node[circle,draw=black, fill=white, inner sep=1pt, align=center] (A1) at (-2,0) {$\ket{a_1}$};
            \node[circle,draw=black, fill=white, inner sep=1pt, align=center] (A2) at (0,0) {$\ket{a_2}$};
            \node[circle,draw=black, fill=white, inner sep=1pt, align=center] (A3) at (2,0) {$\ket{a_3}$};

            \Edge[bend=0,position=left](A1)(A2)
            \Edge[bend=0,position=left](A1)(A3)
            \node[] at (3,0) {$\Box$};

            \node[circle,draw=black, fill=white, inner sep=1pt, align=center] (B1) at (4,1) {$\ket{b_1}$};
            \node[circle,draw=black, fill=white, inner sep=1pt, align=center] (B2) at (4,-1) {$\ket{b_2}$};

            \Edge[bend=0,position=left](B1)(B2)

            \node[] at (5,0) {$=$};

            \node[circle,draw=black, fill=white, inner sep=1pt, align=center] (A1B1) at (6,1) {$\ket{a_1b_1}$};
            \node[circle,draw=black, fill=white, inner sep=1pt, align=center] (A2B1) at (8,1) {$\ket{a_2b_1}$};
            \node[circle,draw=black, fill=white, inner sep=1pt, align=center] (A3B1) at (10,1) {$\ket{a_3b_3}$};
            \node[circle,draw=black, fill=white, inner sep=1pt, align=center] (A1B2) at (6,-1) {$\ket{a_1b_2}$};
            \node[circle,draw=black, fill=white, inner sep=1pt, align=center] (A2B2) at (8,-1) {$\ket{a_2b_2}$};
            \node[circle,draw=black, fill=white, inner sep=1pt, align=center] (A3B2) at (10,-1) {$\ket{a_3b_2}$};

            \Edge[bend=0,position=left](A1B1)(A2B1)
            \Edge[bend=0,position=left](A2B1)(A3B1)
            \Edge[bend=0,position=left](A1B2)(A2B2)
            \Edge[bend=0,position=left](A2B2)(A3B2)
            \Edge[bend=0,position=left](A1B1)(A1B2)
            \Edge[bend=0,position=left](A2B1)(A2B2)
            \Edge[bend=0,position=left](A3B1)(A3B2)

\end{tikzpicture}

%% file: pictures/Dicke_5_1_4.tex
\begin{tikzpicture}

\node[color=black,opacity=1] (1hot) at (0,0.175*1.4) {"1-hot:"};
\node[color=black,opacity=1] (1hot) at (0,2.058*1.4) {"2-hot:"};
\node[color=black,opacity=1] (1hot) at (0,3.942*1.4) {"3-hot:"};
\node[color=black,opacity=1] (1hot) at (0,5.825*1.4) {"4-hot:"};

\node[color=black,opacity=1] (00001) at (1*3,0.175*1.4) {$\scriptstyle\ket{00001}$};
\node[color=black,opacity=1] (00010) at (2*3,0.175*1.4) {$\scriptstyle\ket{00010}$};
\node[color=black,opacity=1] (00100) at (3*3,0.175*1.4) {$\scriptstyle\ket{00100}$};
\node[color=black,opacity=1] (01000) at (4*3,0.175*1.4) {$\scriptstyle\ket{01000}$};
\node[color=black,opacity=1] (10000) at (5*3,0.175*1.4) {$\scriptstyle\ket{10000}$};

\node[color=black,opacity=1] (00011) at (0.5*3,2.058*1.4) {$\scriptstyle\ket{00011}$};
\node[color=black,opacity=1] (00101) at (1.0*3,2.058*1.4) {$\scriptstyle\ket{00101}$};
\node[color=black,opacity=1] (00110) at (1.5*3,2.058*1.4) {$\scriptstyle\ket{00110}$};
\node[color=black,opacity=1] (01001) at (2.0*3,2.058*1.4) {$\scriptstyle\ket{01001}$};
\node[color=black,opacity=1] (01010) at (2.5*3,2.058*1.4) {$\scriptstyle\ket{01010}$};
\node[color=black,opacity=1] (01100) at (3.0*3,2.058*1.4) {$\scriptstyle\ket{01100}$};
\node[color=black,opacity=1] (10001) at (3.5*3,2.058*1.4) {$\scriptstyle\ket{10001}$};
\node[color=black,opacity=1] (10010) at (4.0*3,2.058*1.4) {$\scriptstyle\ket{10010}$};
\node[color=black,opacity=1] (10100) at (4.5*3,2.058*1.4) {$\scriptstyle\ket{10100}$};
\node[color=black,opacity=1] (11000) at (5.0*3,2.058*1.4) {$\scriptstyle\ket{11000}$};

\node[color=black,opacity=1] (00111) at (0.5*3,3.942*1.4) {$\scriptstyle\ket{00111}$};
\node[color=black,opacity=1] (01011) at (1.0*3,3.942*1.4) {$\scriptstyle\ket{01011}$};
\node[color=black,opacity=1] (01101) at (1.5*3,3.942*1.4) {$\scriptstyle\ket{01101}$};
\node[color=black,opacity=1] (01110) at (2.0*3,3.942*1.4) {$\scriptstyle\ket{01110}$};
\node[color=black,opacity=1] (11100) at (2.5*3,3.942*1.4) {$\scriptstyle\ket{11100}$};
\node[color=black,opacity=1] (10011) at (3.0*3,3.942*1.4) {$\scriptstyle\ket{10011}$};
\node[color=black,opacity=1] (10101) at (3.5*3,3.942*1.4) {$\scriptstyle\ket{10101}$};
\node[color=black,opacity=1] (10110) at (4.0*3,3.942*1.4) {$\scriptstyle\ket{10110}$};
\node[color=black,opacity=1] (11001) at (4.5*3,3.942*1.4) {$\scriptstyle\ket{11001}$};
\node[color=black,opacity=1] (11010) at (5.0*3,3.942*1.4) {$\scriptstyle\ket{11010}$};

\node[color=black,opacity=1] (01111) at (1*3,5.825*1.4) {$\scriptstyle\ket{01111}$};
\node[color=black,opacity=1] (10111) at (2*3,5.825*1.4) {$\scriptstyle\ket{10111}$};
\node[color=black,opacity=1] (11011) at (3*3,5.825*1.4) {$\scriptstyle\ket{11011}$};
\node[color=black,opacity=1] (11101) at (4*3,5.825*1.4) {$\scriptstyle\ket{11101}$};
\node[color=black,opacity=1] (11110) at (5*3,5.825*1.4) {$\scriptstyle\ket{11110}$};

\definecolor{XXIII}{RGB}{100,249,84}
\definecolor{IXXII}{RGB}{77,193,65}
\definecolor{IIXXI}{RGB}{57,143,48}
\definecolor{IIIXX}{RGB}{28,71,24}
\definecolor{IIIXI}{RGB}{0,0,221}

\Edge[color=IIIXX,bend=-56.31](00010)(00001)
\Edge[color=IIIXX,bend=-56.31](00101)(00110)
\Edge[color=IIIXX,bend=-56.31](01001)(01010)
\Edge[color=IIIXX,bend=-56.31](01101)(01110)
\Edge[color=IIIXX,bend=-56.31](10001)(10010)
\Edge[color=IIIXX,bend=-56.31](10101)(10110)
\Edge[color=IIIXX,bend=-56.31](11001)(11010)
\Edge[color=IIIXX,bend=-56.31](11101)(11110)

\Edge[color=IIXXI,bend=-56.31](00100)(00010)
\Edge[color=IIXXI,bend=-56.31](00011)(00101)
\Edge[color=IIXXI,bend=-56.31](01010)(01100)
\Edge[color=IIXXI,bend=-56.31](01011)(01101)
\Edge[color=IIXXI,bend=-56.31](10010)(10100)
\Edge[color=IIXXI,bend=-56.31](10011)(10101)
\Edge[color=IIXXI,bend=-56.31](11100)(11010)
\Edge[color=IIXXI,bend=-56.31](11011)(11101)

\Edge[color=IXXII,bend=-56.31](01000)(00100)
\Edge[color=IXXII,bend=-56.31](00101)(01001)
\Edge[color=IXXII,bend=-56.31](00110)(01010)
\Edge[color=IXXII,bend=-56.31](00111)(01011)
\Edge[color=IXXII,bend=-56.31](10100)(11000)
\Edge[color=IXXII,bend=-56.31](10101)(11001)
\Edge[color=IXXII,bend=-56.31](10110)(11010)
\Edge[color=IXXII,bend=-56.31](10111)(11011)

\Edge[color=XXIII,bend=-56.31](10000)(01000)
\Edge[color=XXIII,bend=-56.31](01001)(10001)
\Edge[color=XXIII,bend=-56.31](01010)(10010)
\Edge[color=XXIII,bend=-56.31](01011)(10011)
\Edge[color=XXIII,bend=-56.31](01100)(10100)
\Edge[color=XXIII,bend=-56.31](01101)(10101)
\Edge[color=XXIII,bend=-56.31](01110)(10110)
\Edge[color=XXIII,bend=-56.31](01111)(10111)

\Edge[color=IIIXI](00001)(00011)
\Edge[color=IIIXI](00100)(00110)
\Edge[color=IIIXI](01001)(01011)
\Edge[color=IIIXI](01100)(01110)
\Edge[color=IIIXI](10001)(10011)
\Edge[color=IIIXI](10100)(10110)
\Edge[color=IIIXI](11001)(11011)
\Edge[color=IIIXI](11100)(11110)

\def\shift{3.3}

\draw[XXIII,very thick] (0*\shift,-1.) node[right,xshift=.7cm] {$X_1 X_2\langle -Z_1 Z_2\rangle$} -- (0*\shift+.7,-1.);

\draw[IXXII,very thick] (1*\shift,-1.) node[right,xshift=.7cm] {$X_2 X_3\langle -Z_2 Z_3\rangle$} -- (1*\shift+.7,-1.);

\draw[IIXXI,very thick] (2*\shift, -1.) node[right,xshift=.7cm] {$X_3 X_4\langle -Z_3 Z_4\rangle$} -- (2*\shift+.7,-1.);

\draw[IIIXX,very thick] (3*\shift,-1.) node[right,xshift=.7cm] {$X_4 X_5\langle -Z_4 Z_5\rangle$} -- (3*\shift+.7,-1.);

\draw[IIIXI,very thick] (4*\shift,-1.) node[right,xshift=.7cm] {$X_4\langle -Z_3 Z_5\rangle$} -- (4*\shift+.7,-1.);

\end{tikzpicture}

%% file: pictures/my_figure.tex
\begin{tikzpicture}[scale=1.3]
\foreach \i in {0,...,9} {
   \node[color=black, opacity=1, draw, circle, inner sep=1pt] (\i) at ({360/10 * \i}:2) {\i};
}
  \draw (0) to (4);
  \draw (0) to (7);
  \draw (1) to (4);
  \draw (1) to (5);
  \draw (1) to (6);
  \draw (2) to (4);
  \draw (2) to (5);
  \draw (2) to (6);
  \draw (2) to (7);
  \draw (2) to (9);
  \draw (3) to (4);
  \draw (3) to (5);
  \draw (3) to (6);
  \draw (3) to (8);
  \draw (4) to (5);
  \draw (4) to (7);
  \draw (4) to (8);
  \draw (4) to (9);
  \draw (5) to (6);
  \draw (5) to (7);
  \draw (5) to (8);
  \draw (5) to (9);
  \draw (7) to (8);
  \draw (8) to (9);

\end{tikzpicture}

%% file: pictures/approx_ratio.tex
\begin{tikzpicture}[scale=.8]

\definecolor{darkgray176}{RGB}{176,176,176}
\definecolor{steelblue31119180}{RGB}{31,119,180}

\begin{axis}[
tick align=outside,
tick pos=left,
title={LX-QAOA},
x grid style={darkgray176},
xlabel={depth},
xmin=0, xmax=22,
xtick style={color=black},
y grid style={darkgray176},
ylabel={Approx. ratio},
ymin=0., ymax=1.0,
ytick style={color=black}
]
\addplot [semithick, steelblue31119180, mark=*, mark size=3, mark options={solid}]
table {%
1 0.764996879940149
5 0.810134664613433
9 0.849907050297549
13 0.872256196424733
17 0.900113492076675
21 0.93809272858573
};
\end{axis}

\end{tikzpicture}